\documentclass[a4paper]{article}

\usepackage{amsmath}
\usepackage{amsthm}
\usepackage{amssymb}
\usepackage{hyperref}
\usepackage{xcolor}
\usepackage{enumerate}
\usepackage{thmtools,thm-restate}

\usepackage{tikz}
\usetikzlibrary{arrows,automata,backgrounds,calc,positioning}
\usetikzlibrary{decorations.pathreplacing}
\usetikzlibrary{decorations.pathmorphing}
\tikzstyle{state} = [circle,draw, inner sep = 0, minimum size = 15pt]

\theoremstyle{plain}
\newtheorem{thm}{Theorem}[section]
\newtheorem{prop}[thm]{Proposition}
\newtheorem{defn}[thm]{Definition}
\newtheorem{lem}[thm]{Lemma}
\newtheorem{cor}[thm]{Corollary}
\newtheorem{rem}[thm]{Remark}

\theoremstyle{remark}
\newtheorem{claim}{Claim}
\usepackage{etoolbox}
\AtEndEnvironment{proof}{\setcounter{claim}{0}}
\newtheorem*{claimproof}{Proof}

\theoremstyle{plain}
\newtheorem{corollary-restatable}[thm]{Corollary}
\newtheorem{theorem-restatable}[thm]{Theorem}
\newtheorem{lemma-restatable}[thm]{Lemma}
\newtheorem{proposition-restatable}[thm]{Proposition}

\theoremstyle{definition}
\newtheorem{ex}[thm]{Example}

\renewcommand{\O}{\mathcal{O}}

\newcommand{\N}{\mathbb{N}}

\newcommand{\shift}{\mathrm{shift}}
\newcommand{\cut}{\mathrm{cut}}

\newcommand{\eps}{\epsilon}

\newcommand{\last}{\mathrm{last}}

\newcommand{\dist}{\mathrm{dist}}
\newcommand{\pdist}{\mathrm{pdist}}
\newcommand{\Prob}{\mathrm{Prob}}

\DeclareMathOperator{\Acc}{Acc}
\DeclareMathOperator{\ps}{ps}

\usepackage{authblk}

\title{Sliding window property testing\\ for regular languages}
\author[1]{Moses Ganardi}
\author[1]{Danny Hucke}
\author[1]{Markus Lohrey}
\author[2]{Tatiana Starikovskaya}
\affil[1]{Universit\"{a}t Siegen}
\affil[2]{DI/ENS, PSL Research University, France}
\date{}                     
\begin{document}

\maketitle

\begin{abstract}
We study the problem of recognizing regular languages in a variant of the streaming model of computation, called the sliding window model. In this model, we are given a size of the sliding window $n$ and a stream of symbols. At each time instant, we must decide whether the suffix of length $n$ of the current stream (``the active window'') belongs to a given regular language. 

Recent works \cite{GHKLM18, GHL16} showed that the space complexity of an optimal deterministic sliding window algorithm for this problem is either constant, logarithmic or linear in the window size $n$ and provided natural language theoretic characterizations of the space complexity classes. Subsequently,~\cite{GanardiHL18} extended this result to randomized algorithms to show that any such algorithm admits either constant, double logarithmic, logarithmic or linear space complexity. 

In this work, we make an important step forward and combine the sliding window model with the property testing setting, which results in ultra-efficient algorithms for all regular languages. Informally, a sliding window property tester must accept the active window if it belongs to the language and reject it if it is far from the language. We consider deterministic and randomized sliding window property testers with one-sided and two-sided errors. In particular, we show that for any regular language, there is a deterministic sliding window property tester that uses logarithmic space and a randomized sliding window property tester with two-sided error that uses constant space. 
\end{abstract}


\section{Introduction}
Regular expression search constitutes an important part of many search engines for biological data or code, such as, for example, Elasticsearch Service\footnote{\href{https://www.elastic.co}{https://www.elastic.co}}. In this paper, we consider the following formalization of this problem. We assume to be given an integer $n$, a regular language $L$, and a stream of symbols that we receive one symbol at a time. At each time instant, we have direct access only to the last arrived symbol, and must decide whether the suffix of length $n$ of the current stream (``the active window'') belongs to $L$.

The model described above is a variant of the streaming model and was introduced by Datar et al.~\cite{DatarGIM02}, where the authors proved that the number of $1$'s in a $0/1$-sliding window of size $n$ can be maintained in space $\O(\frac{1}{\eps} \cdot \log^2 n)$ if one allows a multiplicative error of $1\pm \eps$. The motivation for this model of computation is that in many streaming applications, data items are outdated after a certain time, and  the sliding window setting is a simple way to model this. In general, we aim to avoid storing the window content explicitly, and, instead, to work in considerably smaller space, e.g. polylogarithmic space with respect to the window length. For more details on the sliding window model see~\cite[Chapter~8]{Aggarwal07}.
	
The study of recognizing regular languages in the sliding window model was commenced in~\cite{GHKLM18,GHL16}. In~\cite{GHL16}, Ganardi et al.~showed that for every regular language $L$ the optimal space bound for a deterministic sliding window
algorithm is either constant, logarithmic or linear in the window size $n$. In~\cite{GHKLM18}, Ganardi et al.~gave characterizations for these space classes. More formally, they showed that a regular language has a deterministic sliding window algorithm with space $\O(\log n)$ (resp.,  $\O(1)$)  if and only if it is a Boolean combination of so-called 
regular left-ideals and regular length languages (resp., suffix-testable languages and regular length languages).
A subsequent work~\cite{GanardiHL18} studied the space complexity of randomized sliding window algorithms for regular languages. It was shown that for every regular language $L$ the optimal space bound of randomized sliding window algorithm is $\O(1)$, $\O(\log\log n)$, $\O(\log n)$, or~$\O(n)$. Moreover, complete characterizations of these space classes were provided. 

\subsection{Our results}
Previous study implies that even simple languages require linear space in the sliding window model, which gives the motivation to seek for novel approaches in order to achieve efficient algorithms for all regular languages. We take our inspiration from the property testing model introduced by Goldreich et. al~\cite{GoldreichGR98}. In this model, the task is to decide whether the input has a particular property~$P$, or is ``far'' from any input satisfying it. For a function $\gamma: \mathbb{N} \to \mathbb{R}_{\geq 0}$, we say that a word $w$ of length $n$ is $\gamma$-far from satisfying $P$, if the Hamming distance between $w$ and any word $w'$ satisfying $P$ is at least $\gamma(n)$. We will call the function $\gamma(n)$ the Hamming gap of the tester. We must make the decision by inspecting as few symbols of the input as possible, and the time complexity of the algorithm is defined to be equal to the number of inspected symbols. 
The motivation is that when working with large-scale data, accessing a data item is a very time-expensive operation.
The membership problem for a regular language in the property testing model was studied by  Alon et al.~\cite{AlonKNS00} who showed that for every regular language $L$ and every constant $\eps > 0$, there is a property tester with Hamming gap $\gamma(n) = \eps n$ for deciding membership in $L$ that can make the decision by inspecting a random constant-size sample of symbols of the input word. 

In this work, we introduce a class of algorithms called \emph{sliding window property testers}. Informally, at each time moment, a sliding window property tester must accept if the active window has the property $P$ and reject if it is far from satisfying $P$. The space complexity of a sliding window property tester is defined to be all the space used, including the space we need to store information about the input. We consider deterministic sliding window property testers and randomized sliding window property testers with one-sided and two-sided errors (for a formal definition, see Section~\ref{sec:swd_definition}). 
A similar but simpler model of streaming property testers, where the whole stream is considered, was introduced by 
Feigenbaum et al.~\cite{FeigenbaumKSV02}.  Fran{\c{c}}ois et al.~\cite{FrancoisMRS16} continued the study of this model in the context of language membership problems and
came up with a streaming property tester for visibly pushdown languages that uses polylogarithmic space. Note that deciding membership in a regular languages becomes trivial in this model (where
the active window is the whole stream): one can simply simulate a deterministic finite automaton on the stream. What makes the sliding window model more difficult is the fact 
that the oldest symbol in the active window expires in the next step.

While at first sight the only connection between property testers and sliding window property testers is that we must accept the input if it satisfies $P$ and reject if it is far from satisfying $P$, there is, in fact, a deeper link. In particular, the above mentioned result of Alon et al.~\cite{AlonKNS00}  combined with an optimal sampling algorithm for sliding windows~\cite{BravermanOZ12}, immediately yields a $\O(\log n)$-space, two-sided error sliding window property tester with Hamming gap $\gamma(n) = \eps n$ for every regular language. We will improve on this observation. 
Our main contribution are tight complexity bounds for each of the following classes of sliding window property testers for regular languages: 
deterministic sliding window property testers and randomized sliding window property testers with one-sided and two-sided error.

\smallskip
\noindent
\textbf{Deterministic sliding window property testers.} We call a language $L$ \emph{trivial}, if for some constant~$c > 0$ the following holds: For every word $w \in \Sigma^*$ such that $L$ contains a word of length $|w|$, the Hamming distance from $w$ to $L$ is at most $c$.  Every trivial regular language has a constant-space deterministic sliding window property tester with constant Hamming gap (Theorem~\ref{prop-trivial1}). For generic regular languages, we show a deterministic sliding window property tester with constant Hamming gap that uses $\O(\log n)$ space. This is particularly surprising, because for Hamming gap zero (i.e., the exact case) \cite{GanardiHL18} showed a space lower bound of $\Omega(n)$ for generic regular languages. In other words, a constant Hamming gap allows an exponential space improvement. We also show that for {\em non-trivial} regular languages, $\O(\log n)$ space is the best one can hope to achieve, even for Hamming gap $\gamma(n) = \eps n$ (Theorem~\ref{theorem:deterministic_lb}). 

\smallskip
\noindent
\textbf{Randomized sliding window property testers with two-sided error.} Next, we show that for every regular language, there is a randomized sliding window property tester with Hamming gap $\gamma(n) = \eps n$ and two-sided error that uses constant space (Theorem~\ref{theorem:two-sided}). This is an optimal bound and a considerable improvement compared to the tester that can be obtained by combining the property tester of Alon et al.~\cite{AlonKNS00} and an optimal sampling algorithm for sliding windows~\cite{BravermanOZ12}.

\smallskip
\noindent
\textbf{Randomized sliding window property testers with one-sided error.} While our randomized sliding window property tester with two-sided error is optimal, we believe that a two-sided error is a very strong relaxation and to be avoided in some applications. To this end, we study the one-sided error randomized setting. The general landscape for this setting is the most complex: In Theorems~\ref{theorem:one-sided_ub} and \ref{theorem:one-sided_lb}, we show that for every regular language $L$, the space complexity of an optimal randomized sliding window property tester with one-sided error is either $\O(1)$, $\O(\log \log n)$, or $\O(\log n)$, and we provide characterizations of these complexity classes.

In order to show our upper bound results, we demonstrate novel combinatorial properties of automata and regular languages and develop new streaming techniques, such as probabilistic counters, which can be of interest on their own. To show the lower bound results, we introduce a new methodology, which could potentially simplify further establishments of lower bounds in string processing tasks in the streaming setting: Namely, we view the testers as nondeterministic automata, and study their behaviour. 

\subsection{Related work}
The results above assume that the regular language admits a constant-space description and we will follow the same assumption in this work. Currently, there are few studies on the dependency of the complexity of sliding window algorithms on the size of the language description.
On the negative side, Ganardi et al.~\cite{GHKLM18} showed that there are regular languages such that any sliding window algorithm that achieves logarithmic space (in the window size) depends exponentially on the automata size.

On the positive side, there is an extensive study of the pattern matching problem and its variants that gives sub-exponential upper bounds for a class of (very simple) regular languages. In this problem, we are given a pattern and a streaming text $T$, and at each moment we must decide if the active window is equal to the pattern. This problem and its generalisations have been studied in~\cite{BG14,CFP+15,CFP+16,CKP19,CS16,GKP16,GKP18,GP17,PP09,S17}. 

Similar to regular languages, we can ask whether the current active window belongs to a given context-free language. This question was studied 
in~\cite{BabuLRV13,JN14,KrebsLS11,MagniezMN14} for the model where the active window is the complete stream
and in~\cite{G19,GanardiJL18} for the sliding-window model.

\section{Sliding window property tester}\label{sec:swd_definition}
We fix a finite alphabet $\Sigma$ for the rest of the paper.
We denote by $\Sigma^*$ the set of all words over $\Sigma$
and by $\Sigma^n$ the set of words over $\Sigma$ of length $n$.
The empty word is denoted by $\lambda$.
Let $w$ be a word.
We say that $v$ is a {\em prefix (suffix)} of $w$
if $w = xv$ ($w = vx$) for some word $x$.
We say that $v$ is a {\em factor} of $w$
if $w = xvy$ for some words $x,y$.
 The {\em Hamming distance} between two words $u = a_1 \cdots a_n$ and $v = b_1 \cdots b_n$ of equal length
is the number of positions where $u$ and $v$ differ, i.e. $\dist(u,v) = |\{ i :  a_i \neq b_i\}|$.
The distance of a word $u$ to a language $L$ is defined as
$\dist(u,L) = \inf \{ \dist(u,v) : v \in L \} \in \N \cup \{\infty\}$.

A {\em deterministic finite automaton} (DFA) is a tuple $A = (Q,\Sigma,q_0,\delta,F)$ where 
$Q$ is a finite set of states, $\Sigma$ is the input alphabet, $q_0$ is the initial state,
$\delta : Q \times \Sigma \to Q$ is the transition mapping and $F \subseteq Q$ is the set of 
final states. We extend $\delta$ to a mapping $\delta : Q \times \Sigma^* \to Q$ inductively in the usual way:
$\delta(q,\lambda) = q$ and $\delta(q,aw) = \delta( \delta(q,a),w)$.
The language accepted by $A$ is $L(A) = \{ w \in \Sigma^* : \delta(q_0,w) \in F \}$. 
A language is {\em regular} if it is accepted by a DFA.
For more background in automata theory see \cite{HoUl79}.

A {\em stream} is a word $a_1 a_2 \cdots a_m$ over $\Sigma$. A {\em sliding window algorithm}
is a family $\mathcal{A}=(A_n)_{n \geq 0}$ of streaming algorithms.
Given a window size $n \in \N$ and an input stream $a_1 a_2 \cdots a_m \in \Sigma^*$
the algorithm $A_n$ reads the stream symbol by symbol from left to right and thereby updates its memory content. 
After reading a prefix $a_1 \cdots a_t$ ($0 \le t \le m$) the algorithm is required to compute an output value that depends on the {\em active window} $\last_n(a_1 \cdots a_t) = a_{t-n+1}\cdots a_t$ at time $t$. For convenience, for $i < 0$ we define $a_i = \square$ where $\square\in\Sigma$ is an arbitrary fixed symbol. In other words, we assume an initial window $\square^n$ that is active at time $t=0$.
We consider {\em deterministic sliding window algorithms} (where every $A_n$ can be viewed as a DFA) and {\em randomized sliding window algorithms}
(where every $A_n$ can be viewed as a probabilistic finite automaton in the sense of Rabin \cite{Rabin63}). 
In the latter case, $A_n$ updates in each step its memory content according to a probability distribution that
depends on the current memory content and the current input symbol.
Let $\gamma : \N \to \mathbb{R}_{\geq 0}$ be a function such that $\gamma(n) \leq n$ for all $n \in \N$ and let $\alpha, \beta$ be probabilities.

\begin{defn}
A {\em deterministic sliding window (property) tester} for a language $L$ with Hamming gap $\gamma(n)$ is a deterministic sliding window algorithm $\mathcal{A}=(A_n)_{n \geq 0}$ such that for every input stream $w \in \Sigma^*$ and every window size $n$ the following properties hold:
\begin{itemize}
\item if $\last_n(w) \in L$, then $A_n$ accepts;
\item if $\dist(\last_n(w),L) > \gamma(n)$, then $A_n$ rejects.
\end{itemize}
\end{defn}

\begin{defn}
A {\em randomized sliding window (property) tester} for a language $L$ with Hamming gap $\gamma(n)$ and error $(\alpha,\beta)$ is a randomized sliding window algorithm $\mathcal{A}=(A_n)_{n \geq 0}$ such that for every input stream $w \in \Sigma^*$ and every window size $n$ the following properties hold:
\begin{itemize}
\item if $\last_n(w) \in L$, then $A_n$ accepts with probability at least $1-\alpha$;
\item if $\dist(\last_n(w),L) > \gamma(n)$, then $A_n$ rejects with probability at least $1-\beta$.
\end{itemize}
We say that $\mathcal{A}$ has \emph{one-sided error}  if $\mathcal{A}$ has error $(0,1/2)$ and \emph{two-sided error} if $\mathcal{A}$ has error $(1/3,1/3)$.
\end{defn}
Notice that our definition is non-uniform since we allow an arbitrary algorithm $A_n$ for each window size $n$. 
If the window size is not specified, then it is implicitly universally quantified. 
The space consumption of $\mathcal{A}$ is the mapping $s(n)$, where $s(n)$ is the space consumption of $A_n$, i.e., the maximal number of bits stored by $A_n$ while reading any input stream.
We can assume that $s(n) \in \O(n)$ since $A_n$ can store the active window in $\O(n)$ bits.
The goal is to devise algorithms which only use $o(n)$ space. Using  probability amplification (similar to \cite{GanardiHL18}) one can replace the error probability $1/3$ in the two-sided error setting (resp. $1/2$ in the one-sided error setting) by any probability $p<1/2$ (resp. $p<1$). This influences the space complexity only by a constant factor.
The case of Hamming gap $\gamma(n) = 0$ corresponds to exact membership testing to $L$ which was studied in \cite{GHKLM18,GHL16,GanardiHL18}. In this paper, we focus on the two cases $\gamma(n) = c$ for some constant $c > 0$ and $\gamma(n) = \eps n$ for some $\eps > 0$. 

\begin{rem} \label{rem-finite-union}
Assume that $L = \bigcup_{i=1}^k L_i$ and that for every $1 \leq i \leq k$ there exists a randomized sliding window tester for $L_i$ with 
Hamming gap $\gamma(n)$ and error $(\alpha,\beta)$ that uses space $s_i(n)$. We can combine these testers into 
a sliding window tester for $L$ with Hamming gap $\gamma(n)$ and error $(\alpha,\beta)$ that uses space 
$\O(\sum_{i=1}^k s_i(n))$: First, using probability amplification, we reduce the error of each given sliding window tester
to $(\alpha/k,\beta/k)$.
Then we run the sliding window testers for $L_i$ in parallel and accept if and only if one of them accepts.
\end{rem}



\section{Main results}

Our first main result is a deterministic logspace sliding window tester for every regular language,
together with a matching lower bound for so-called {\em nontrivial} regular languages (defined below).
\begin{restatable}[deterministic setting, upper bound]{theorem-restatable}{determub}
\label{theorem:deterministic_ub}
For every regular language $L$, there exists a deterministic sliding window tester
for $L$ with constant Hamming gap which uses $\O(\log n)$ space.
\end{restatable}

\begin{restatable}[deterministic setting, lower bound]{theorem-restatable}{determlb}
\label{theorem:deterministic_lb}
For every non-trivial regular language $L$, there exist $\eps > 0$ and infinitely many window sizes $n \in \N$ on which every deterministic sliding window tester for $L$ with Hamming gap $\eps n$ uses space $\Omega(\log n)$.
\end{restatable}

Here the notion of (non-)trivial languages is defined as follows:
Let $\gamma : \mathbb{N} \to  \mathbb{R}_{\geq 0}$ be a mapping such that $\gamma(n) \leq n$ for all $n \geq 0$.
A language is $L \subseteq \Sigma^*$ is {\em $\gamma$-trivial}
if there exists a number $n_0$ such that for all $n \ge n_0$ with $L \cap \Sigma^n \neq \emptyset$
and all $w \in \Sigma^n$ we have $\dist(w,L) \le \gamma(n)$.
If $\gamma(n) \in \mathcal{O}(1)$, we say that $L$ is {\em trivial}.
Note that Alon et al.~\cite{AlonKNS00} call a language $L$ trivial
if $L$ is $(\eps n)$-trivial for all $\eps > 0$ according to our definition.
In fact, we will prove that both definitions coincide for regular languages (Corollary~\ref{triv-alon}).

Next we consider randomized sliding window property testers. 
Our second main result is a constant-space randomized sliding window property tester with two-sided error
for any regular language.

\begin{restatable}[two-sided error randomized setting, upper bound]{theorem-restatable}{twoside}
\label{theorem:two-sided}
For every regular language $L$ and every $\eps > 0$,
there exists a randomized sliding window tester for $L$
with two-sided error and Hamming gap $\gamma(n) = \eps n$ that uses space $\O(1/\eps)$. 
\end{restatable}

While the randomized setting with two-sided error allows efficient testers, we find that allowing a two-sided error is a very strong relaxation. To this end, we study the randomized setting with one-sided error. In this setting, only a small class of regular languages admits sliding window testers working in space $o(\log n)$.
A language $L \subseteq \Sigma^*$ is {\em suffix-free} if $xy \in L$ and $x \neq \lambda$ imply $y \notin L$.

\begin{restatable}[one-sided error randomized setting, upper bound]{theorem-restatable}{onesideub}
\label{theorem:one-sided_ub}
 If $L$ is a finite union of trivial regular languages and suffix-free regular languages, then there exists a randomized sliding window tester for $L$ with one-sided error and constant Hamming gap which uses $\O(\log \log n)$ space.
\end{restatable}

\begin{restatable}[one-sided error randomized setting, lower bound]{theorem-restatable}{onesidelb}
\label{theorem:one-sided_lb}
Let $L$ be a regular language.
\begin{itemize}
\item If $L$ is not a finite union of trivial regular languages and suffix-free regular  languages, there exist $\eps > 0$ and infinitely many window sizes $n$ on which every randomized sliding window tester for $L$ with one-sided error and Hamming gap $\eps n$ uses space $\Omega(\log n)$.
\item If $L$ is non-trivial, then there exist $\eps > 0$ and infinitely many window sizes $n$
on which every sliding window tester for $L$ with one-sided error and Hamming gap $\eps n$
uses space $\Omega(\log \log n)$.
\end{itemize}
\end{restatable}

\medskip
We provide the proofs of Theorem~\ref{theorem:deterministic_ub},~\ref{theorem:two-sided}, and~\ref{theorem:one-sided_ub} in
Sections~\ref{sec:deterministic_ub},~\ref{sec:two-sided}, and~\ref{sec:one-sided_ub}, respectively. The proofs of Theorems~\ref{theorem:deterministic_lb} and~\ref{theorem:one-sided_lb} can be found in Section~\ref{section-lower-bounds}. We would like to emphasize that the lower bounds from Section~\ref{section-lower-bounds} are stronger than those stated in Theorems~\ref{theorem:deterministic_lb} and~\ref{theorem:one-sided_lb}. More precisely,  we show space lower bounds for nondeterministic and co-nondeterministic sliding window testers; see Section~\ref{section-lower-bounds} for definitions. 

\section{Trivial languages} \label{sec-trivial}

Let us start by analyzing trivial regular languages.
The reason we introduce trivial languages the way we do (and a justification to call them ``trivial'') is stated in the following theorem:

\begin{restatable}{theorem-restatable}{theoremtrivial}
\label{prop-trivial1}
If $L$ is a trivial language (not necessarily regular), then there is a deterministic sliding window tester for $L$ with constant Hamming gap which uses constant space. The converse is also true: If for a language $L$ there is a deterministic constant-space sliding window tester with Hamming gap $\gamma(n)$, then there exists a constant $c$ such that $L$ is $(\gamma(n)+c)$-trivial. 
\end{restatable}
\begin{proof}
Assume first that $L$ is trivial. Let $n \in \N$ be a window size. If $L\cap \Sigma^n=\emptyset$, then the algorithm always rejects, which is obviously correct since any active window of length $n$ has infinite Hamming distance to $L$. Otherwise, the algorithm always accepts. In this case, we use the fact that $L$ is trivial, i.e., there is a constant $c$ such that the Hamming distance between an arbitrary active window of length $n$ and $L$ is at most~$c$.

We now show the converse statement. Let $\mathcal{A} = (A_n)$ be a deterministic sliding window tester for $L$ with Hamming gap $\gamma(n)$ which uses constant space.
Assume that every $A_n$ works on at most $s$ bits for a constant $s$.
Let $N \subseteq \mathbb{N}$ be the set of all $n$ such that $L \cap \Sigma^n \neq \emptyset$.
Note that every $A_n$ with $n \in N$ can be viewed as a DFA with at most $2^{s+1}$ states
that accepts a non-empty language.
The number of DFAs of size at most $2^{s+1}$ over the input alphabet $\Sigma$ is bounded by a fixed constant $d$ (up to isomorphism).
Hence, at most $d$ different DFAs can appear in the list $(A_n)_{n \in N}$. We therefore can choose 
numbers $n_1 < n_2 < \cdots < n_e$ from $N$ with $e \leq d$ such that
for every $n \in N$ there exists a unique $n_i \leq n$ with $A_n = A_{n_i}$
(here and in the following we do not distinguish between isomorphic DFAs). 
Let us choose for every $1 \le i \le e$ a word $u_i \in L$ of length $n_i$.
Now take any $n \in N$. Assume that $A_n = A_{n_i}$ where $n_i \leq n$. Consider any word $u \in \Sigma^* u_i$.
Since $\last_{n_i}(u) = u_i \in L$, $A_{n_i}$ has to accept $u$. Hence, $A_n$ accepts all words from $\Sigma^* u_i$.
In particular, for every word $x$ of length $n-n_i$, $A_n$ accepts $xu_i$. This implies that
$\dist(x u_i,L) \leq \gamma(n)$ for all $x \in \Sigma^{n-n_i}$.
Recall that this holds for all $n \in N$ and that $N$ is the set of all lengths realized by $L$.
Hence, if we define $c := \max\{n_1, \ldots, n_e\}$ (which is a constant that only depends on our deterministic sliding window tester),
then every word $w$ of length $n \in N$ has Hamming distance at most $\gamma(n)+c$ from a word in $L$. Therefore $L$ is $(\gamma+c)$-trivial. 
\end{proof}

In the rest of the section we show that every nontrivial regular language $L$ is already not 
$\eps n$-trivial for some $\eps>0$. For this we first show some auxiliary results that will be also
used in Section~\ref{section-lower-bounds}.
Given $i,j \ge 0$ and a word $w$ of length at least $i+j$ we define
$\cut_{i,j}(w) = y$ such that $w = xyz$, $|x| = i$ and $|z| = j$.
If $|w| < i+j$, then $\cut_{i,j}(w)$ is undefined.
For a language $L$ we define the {\em cut-language} $\cut_{i,j}(L) = \{ \cut_{i,j}(w) \mid w \in L \}$.

\begin{lem}
	If $L$ is regular, then there are finitely many languages $\cut_{i,j}(L)$.
\end{lem}

\begin{proof}
Let $A = (Q,\Sigma,q_0,\delta,F)$ be a DFA for $L$.
Given $i,j \ge 0$, let $I$ be the set of states reachable from $q_0$ via $i$ symbols
and let $F'$ be the set of states from which $F$ can be reached via $j$ symbols.
Then the nondeterministic finite automaton $(Q,\Sigma,I,\delta,F')$ recognizes $\cut_{i,j}(L)$
(see Section~\ref{section-lower-bounds} for the definition of nondeterministic finite automata).
Since there are at most $2^{2|Q|}$ such choices for $I$ and $F'$,
the number of languages of the form $\cut_{i,j}(L)$ must be finite.
\end{proof}

A language $L$ is a {\em length language} if for all $n \in \N$
either $\Sigma^n\subseteq L$ or $\Sigma^n\cap L= \emptyset$.

\begin{lem} \label{cuts}
	If $\cut_{i,j}(L)$ is a length language for some $i,j \ge 0$, then $L$ is trivial.
\end{lem}

\begin{proof}
	Assume that $\cut_{i,j}(L)$ is a length language.
	Let $n \in \N$ such that $L \cap \Sigma^n \neq \emptyset$
	and $n \ge i+j$. We claim that $\dist(w,L) \le i+j$ for all $w \in \Sigma^n$.
	Let $w \in \Sigma^n$ and $w' \in L \cap \Sigma^n$. Then $\cut_{i,j}(w') \in \cut_{i,j}(L)$
	and hence also $\cut_{i,j}(w) \in \cut_{i,j}(L)$.
	Therefore there exist $x \in \Sigma^i$ and $z \in \Sigma^j$ such that $x \, \cut_{i,j}(w) \, z \in L$
	satisfying $\dist(w, x \, \cut_{i,j}(w) \, z) \le i+j$.
\end{proof}

The {\em restriction} of a language $L$ to a set of lengths $N \subseteq \N$
is $L|_N = \{ w \in L : |w| \in N \}$.
A language~$L$ {\em excludes a word $w$ as a factor} if $w$ is not a factor of any word in $L$.
A simple but important observation is that if $L$ excludes $w$ as a factor
and $v$ contains $k$ disjoint occurrences of $w$,
then $\dist(v,L) \ge k$: If we change at most $k-1$ many symbols in $v$, then the resulting
word $v'$ must still contain $w$ as a factor and hence $v' \notin L$. 

\begin{prop}
	\label{prop-nontriv}
	Let $L$ be regular. If $\cut_{i,j}(L)$ is not a length language for all $i, j \ge 0$,
	then $L$ has an infinite restriction $L|_N$ to an arithmetic progression $N = \{ a+bn \mid n \in \N \}$
	which excludes a factor.
\end{prop}

\begin{proof}
	First notice that $\cut_{i,j}(L)$ determines $\cut_{i+1,j}(L)$ and $\cut_{i,j+1}(L)$: we have
	$\cut_{i+1,j}(L) = \{ w \mid \exists a \in \Sigma : aw \in \cut_{i,j}(L) \}$ and similarly for $\cut_{i,j+1}(L)$.
	Since the number of cut-languages $\cut_{i,j}(L)$ is finite there exist numbers $i \ge 0$ and $d > 0$ such that
	$\cut_{i,0}(L) = \cut_{i+d,0}(L)$. Hence, we have $\cut_{i,j}(L) = \cut_{i+d,j}(L)$ for all $j \ge 0$. 
	By the same argument, there exist numbers $j \ge 0$ and $e > 0$ such that
	$\cut_{i,j}(L) = \cut_{i,j+e}(L) = \cut_{i+d,j}(L) = \cut_{i+d,j+e}(L)$, which implies
	$\cut_{i,j}(L) = \cut_{i,j+h}(L) = \cut_{i+h,j}(L) = \cut_{i+h,j+h}(L)$ for some $h>0$ (we can take $h = ed$).
	This implies that $\cut_{i,j}(L)$ is closed under removing prefixes and suffixes of length $h$.
	
	By assumption $\cut_{i,j}(L)$ is not a length language,
	i.e. there exist words $y' \in \cut_{i,j}(L)$ and $y \notin \cut_{i,j}(L)$
	of the same length $k$.
	Let $N = \{ k + i + j + hn \mid n \in \N \}$.
	For any $n \in \N$ the restriction $L|_N$ contains a word of length $k+i+j+hn$
	because $y' \in \cut_{i,j}(L) = \cut_{i+hn,j}(L)$.
	This proves that $L|_N$ is infinite.
	
	Let $u$ be an arbitrary word which contains for every remainder $0 \le r \le h-1$
	an occurrence of $y$ as a factor starting at a position which is congruent to $r$ mod $h$.
	We claim that $L|_N$ excludes $a^i u a^j$ as a factor where $a$ is an arbitrary symbol.
	Assume that there exists a word $w \in L|_N$ which contains $a^i u a^j$ as a factor.
	Then $\cut_{i,j}(w)$ contains $u$ as a factor, has length $k+hn$ for some $n \ge 0$,
	and belongs to $\cut_{i,j}(L)$.
	Therefore $\cut_{i,j}(w)$ also contains $h$ many occurrences of $y$, one per remainder $0 \le r \le h-1$.
	Consider the occurrence of $y$ in $\cut_{i,j}(w)$ which starts at a position which is divisible by $h$,
	i.e. we can factorize $\cut_{i,j}(w) = xyz$ such that $|x|$ is a multiple of $h$.
	Since $\cut_{i,j}(w)$ has length $k+hn$ also $|z|$ is a multiple of~$h$.
	Therefore $y \in \cut_{i+|x|,j+|z|}(L) = \cut_{i,j}(L)$, which is a contradiction.
\end{proof}

\begin{cor}
\label{triv-alon}
If $L$ is a nontrivial regular language, then there exists $\eps > 0$
such that $L$ is not $\eps n$-trivial.
\end{cor}

\begin{proof}
	Let $L$ be nontrivial and regular.
	By Lemma~\ref{cuts} and Proposition~\ref{prop-nontriv}
	there exists an infinite restriction $L|_N$ of $L$ which excludes a factor $w$.
	Hence if $n \in N$ and $v$ is any word of length $n$, which contains at least $\lfloor n/|w| \rfloor$
	many disjoint occurrences of $w$, then $\dist(v,L) \ge \lfloor n/|w| \rfloor$, which proves the claim.
\end{proof}

\section{More background on automata} \label{prel}

\subparagraph{Right-deterministic finite automata.}
For Section~\ref{sec-rand},
it is convenient to work with DFAs which read the input word from right to left.
A {\em right-deterministic finite automaton (rDFA)} is a tuple $B = (Q,\Sigma,F,\delta,q_0)$,
where $Q$, $\Sigma$, $q_0$ and $F$ are as in a DFA, and $\delta \colon \Sigma \times Q \to Q$
is the transition function. We extend $\delta$ to a mapping $\delta : Q \times \Sigma^* \to Q$ analogously to DFAs:
$\delta(q,\lambda) = q$ and $\delta(q,wa) = \delta( \delta(q,a),w)$. The regular language recognized by the rDFA $B$ is $L(B) = \{ w \in \Sigma^* : \delta(w,q_0) \in F \}$.
A run from $p_0 \in Q$ to $p_n \in Q$ on a word $x = a_n \cdots a_2 a_1 \in \Sigma^*$
is a sequence $\pi = (p_n,a_n, p_{n-1}, \dots, p_2,a_2,p_1,a_1,p_0)$
such that $p_i = \delta(a_i,p_{i-1})$ for all $1 \le i \le n$. The {\em length} of $\pi$ is $|\pi| = n$.
We visualize $\pi$ in the form
\[
	\pi \colon p_n \xleftarrow{a_n} p_{n-1} \xleftarrow{a_{n-1}} \cdots \xleftarrow{a_2} p_1 \xleftarrow{a_1} p_0 .
\]
If $p_n \in F$, then $\pi$ is an \emph{accepting run}.  A run of length $1$ is a {\em transition}.
If $\pi$ is a run from $p$ to $q$ on a word $v$,
and $\rho$ is a run from $q$ to $r$ on a word $u$,
then $\rho\pi$ denotes the unique run from $p$ to $r$ on $uv$.
We denote by $\pi_{w,q}$ the unique run on $w$ from $q$.

\subparagraph{Strongly connected graphs.}
With a DFA $A = (Q,\Sigma,q_0,\delta,F)$ we associate the directed graph $(Q,E)$
with edge set $E = \{ (p,\delta(p,a)) \mid p \in Q, a \in \Sigma\}$.
Similarly, with an rDFA $A = (Q,\Sigma,F,\delta,q_0)$  we associate the directed graph $(Q,E)$
with edge set $E = \{ (p,\delta(a,p)) \mid p \in Q, a \in \Sigma\}$.
Let $A$ be a DFA or an rDFA.
Two states $p,q$ in $A$ are {\em strongly connected} if there exists
a path in $(Q,E)$ from $p$ to $q$, and vice versa.
The {\em strongly connected components (SCCs)} of $A$ with state set $Q$
are the maximal subsets $C \subseteq Q$
in which all states $p,q \in C$ are strongly connected.
A state $q \in Q$ is {\em transient} if there exists no nonempty path from $q$ to $q$.
An SCC $C$ is {\em transient} if it only contains a single transient state.
There is a natural partial order on the SCCs, called the 
{\em SCC-ordering}, where the SCC $C_1$ is smaller than the SCC $C_2$ if there
exists a path in $(Q,E)$ from a state in $C_1$ to a state in $C_2$.

The following combinatorial result from \cite{AlonKNS00} will be used in this paper.
Consider a directed graph $G = (V,E)$. 
The period of $G$ is the greatest common divisor of all cycle lengths in $G$. If $G$ is acyclic
we define the period to be $\infty$.

\begin{lem}[c.f.~\cite{AlonKNS00}] \label{lemma-alon}
Let $G = (V, E)$ be a strongly connected directed graph with $E \neq \emptyset$ and finite period $g$. 
Then there exist a partition $V = \bigcup_{i=0}^{g-1} V_i$ and a constant $m(G) \leq  3|V|^2$ with the following properties:
\begin{itemize}
\item For every $0 \leq i,j \leq g-1$ and for every $u \in V_i$, $v \in V_j$ the length of every directed path from $u$ to 
$v$ in $G$ is congruent to $j-i$ modulo $g$.
\item For every $0 \leq i,j \leq g-1$, for every $u \in V_i$, $v \in V_j$ and every 
integer $r \geq m(G)$, if $r$ is congruent to $j-i$ modulo $g$, then there exists a directed path from $u$ to $v$ in $G$ of length $r$.
\end{itemize}
\end{lem}

If $G = (V,E)$ is strongly connected with $E \neq \emptyset$ and finite period $g$,
and $V_0, \ldots, V_{g-1}$ satisfy the properties from Lemma~\ref{lemma-alon}, then
we define the {\em shift} from $u \in V_i$ to $v \in V_j$
by
\begin{equation}\label{shift}
\shift(u,v) = j-i \pmod g \in \{0, \dots, g-1\}.
\end{equation}
Notice that this definition is independent of the partition $\bigcup_{i=0}^{g-1} V_i$ 
since any path from $u$ to $v$ has length $\ell \equiv \shift(u,v) \pmod g$
by Lemma~\ref{lemma-alon}.  Also note that $\shift(u,v) + \shift(v,u) \equiv 0 \pmod g$.
In the following let $g(C)$ denote the period of the SCC $C$.

\begin{restatable}[Uniform period]{lemma-restatable}{uniformperiod}
\label{lem:uni-per}
For every regular language $L$ there exists an rDFA $A$ for $L$ and a number~$g$ such that every non-transient SCC $C$ in $A$ has period $g(C) = g$.
\end{restatable}
\begin{proof}
Let $B = (Q,\Sigma,F,\delta,q_0)$ be an rDFA for $L$.
Let $g$ be the product of all periods $g(C)$ over all non-transient SCCs $C$.
As usual, we consider $\mathbb{Z}_g = \{0,\ldots,g-1\}$ with arithmetic operations modulo $g$.
Then $A = B \times \mathbb{Z}_g = (Q \times \mathbb{Z}_g, \Sigma, F \times \mathbb{Z}_g, \delta', (q_0,0))$,
where for all $(p,i) \in Q \times \mathbb{Z}_g$ and $a \in \Sigma$ we set
\[
	\delta'(a,(p,i)) = \begin{cases}
	(\delta(a,p),i+1), & \text{if } p \text{ and } \delta(a,p) \text{ are strongly connected}, \\
	(\delta(a,p),0), & \text{otherwise.}
	\end{cases}
\]
Clearly, $A$ is equivalent to $B$.
We show that every non-transient SCC of $A$ has period $g$.
The non-transient  SCCs of $A$ are the sets $C \times  \mathbb{Z}_g$,
where $C$ is a non-transient  SCC of $B$.
Let $C$ be a non-transient  SCC of $B$.
Clearly, every cycle length in $C \times \mathbb{Z}_g$ is a multiple of $g$.
Moreover, by Lemma~\ref{lemma-alon} the SCC $C$ contains a cycle of length $k \cdot g(C)$
for every sufficiently large $k \in \N$ ($k \geq m(C)$ suffices).
Since $g$ is a multiple of $g(C)$, $C$ also contains a cycle of length 
$k \cdot g$ for every sufficiently large $k$. But every such cycle induces
a cycle of the same length $k \cdot g$ in $C \times \mathbb{Z}_g$.
Hence, there exist primes $p_1 \neq p_2$ such that $p_1$ and $p_2$ are not divisors of $g$ and 
$C \times \mathbb{Z}_g$ contains cycles of length $p_1 \cdot g$ and $p_2 \cdot g$.
It follows that the period of  $C \times \mathbb{Z}_g$ divides $\gcd(p_1 \cdot g, p_2 \cdot g) = g$.
This proves that the period of $C \times \mathbb{Z}_g$ is exactly~$g$.
\end{proof}

\section{Upper bounds}
\label{sec-rand}
In this section we provide proofs of Theorems~\ref{theorem:deterministic_ub},~\ref{theorem:two-sided}, and~\ref{theorem:one-sided_ub} that give upper bounds for deterministic and (one-sided and two-sided error)  randomized sliding window testers. All algorithms in this section satisfy the stronger property that words with large prefix distance are rejected by the algorithm with high probability (probability one in the deterministic setting). The {\em prefix distance} between words $u = a_1 \cdots a_n$ and $v = b_1 \cdots b_n$ is 
$\pdist(u,v) = \min\{ i \in \{0, \dots, n\} : a_{i+1} \cdots a_n = b_{i+1} \cdots b_n \}$.
Clearly, we have $\dist(u,v) \le \pdist(u,v)$. We extend the definition to languages: for a language $L$, let $\pdist(u,L) = \min \{ \pdist(u,v) : v \in L \}$. The prefix distance between two runs $\pi = (q_0, a_1, \dots, q_{n-1},a_n, q_n)$ and $\rho = (p_0, b_1, \dots, p_{n-1},b_n, p_n)$ is defined as $\pdist(\pi,\rho) = \min\{ i \in \{0, \dots, n\} : (q_{i},a_{i+1}, \dots, q_{n-1}, a_n,q_n) = (p_{i},b_{i+1}, \dots, p_{n-1},b_n,p_n) \}$.

\subparagraph*{Path summaries.} \label{sec-ps}
We start by recalling the notion of a path summary from~\cite{GHKLM18}, where 
it was used in order to prove a logspace upper bound for regular left-ideals 
(in the exact setting where the Hamming gap is zero). For the rest of Section~\ref{sec-rand} we fix a regular language $L \subseteq \Sigma^*$
and an rDFA $B = (Q,\Sigma,F,\delta,q_0)$ which recognizes $L$. By Lemma~\ref{lem:uni-per}, we can assume that every non-transient SCC $C$ of $B$ has period $g(C) = g$.
Consider a run $\pi = (p_n,a_n,\dots,a_1,p_0)$ on $x = a_n \cdots a_1$. If all states $p_n, \dots, p_0$ are contained in a single SCC we call $\pi$ {\em internal}. We can decompose $\pi = \pi_m \tau_{m-1} \pi_{m-1} \cdots \tau_1 \pi_1$, where each $\pi_i$ is a possibly empty internal run and each $\tau_i$ is a single transition connecting two distinct SCCs. We call this unique factorization the {\em SCC-factorization} of $\pi$, which is illustrated in Figure~\ref{fig:scc-fact}.
The {\em path summary} of $\pi$ is
\[
	\ps(\pi) = ( |\pi_m|,q_m ) (|\tau_{m-1} \pi_{m-1}|,q_{m-1} ) \cdots (|\tau_2 \pi_2|,q_2 ) ( |\tau_1 \pi_1|,q_1 ),
\]
where $q_i$ is the first state in $\pi_i$ ($1 \leq i \leq m$). 
Note that $m$ is bounded by the number of states of $B$, which is a constant in our setting. Hence, a path summary can be stored with $\O(\log |\pi|)$ bits.

\begin{figure}[t]
\centering
\begin{tikzpicture}
  \tikzset{every edge/.style={draw,->,>=stealth'}}
  \tikzstyle{small} = [circle,draw=black,fill=black,inner sep=.5mm]
  
  \node[small,label=below:$q_1$] (q1) {};
  
  \node[small, left = 4em of q1] (p1) {};
  \node[small, label=below:$q_2$,left = 1em of p1] (q2) {};
  \node[small, left = 4em of q2] (p2) {};
  \node[small, label=below:$q_3$,left = 1em of p2] (q3) {};

  \node[small, left = 7em of q3] (pm1) {};
  \node[small, label=below:$q_m$,left = 1em of pm1] (qm) {};
  \node[small,left = 4em of qm] (pm) {};
  
  \draw (q1) edge[decorate,decoration={snake,amplitude=1pt}] node[above] {\footnotesize $\pi_1$} (p1);
  \draw (p1) edge node[above] {\footnotesize $\tau_1$} (q2);
  \draw (q2) edge[decorate,decoration={snake,amplitude=1pt}] node[above] {\footnotesize $\pi_2$} (p2);
  \draw (p2) edge node[above] {\footnotesize $\tau_2$} (q3);

  \draw (pm1) edge node[above] {\footnotesize $\tau_{m-1}$} (qm);
  \draw (qm) edge[decorate,decoration={snake,amplitude=1pt}] node[above] {\footnotesize $\pi_m$} (pm);
  
  \draw (q3) edge[-,dotted] (pm1);
\end{tikzpicture}

\caption{The SCC-factorization of a run}\label{fig:scc-fact}
\end{figure}
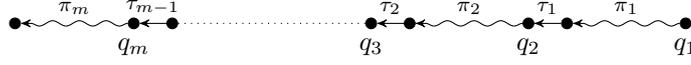

\subparagraph*{Periodic acceptance sets.} \label{sec-periodic}

For $a \in \N$ and $X \subseteq \N$ we use the standard notation
$X + a = \{a+x : x \in X\}$.
For a state $q \in Q$ we define 
\[ 
\Acc(q) = \{ n \in \N : \exists w \in \Sigma^n : \delta(w,q) \in F \}.
\]
A set $X \subseteq \N$ is {\em eventually $d$-periodic},
where $d \ge 1$ is an integer,
if there exists a {\em threshold} $t \in \N$ such that for all $x \ge t$
we have $x \in X$ if and only if $x + d \in X$.
If $X$ is eventually $d$-periodic for some $d \ge 1$, then $X$ is {\em eventually periodic}.

\begin{restatable}{lemma-restatable}{eventuallyperiodic}
\label{lem:acc-per}
For every $q \in Q$ the set $\Acc(q)$ is eventually $g$-periodic.
\end{restatable}
\begin{proof}
		It suffices to show that for all $0 \le r \le g-1$
        the set $S_r = \{ i \in \N : r + i \cdot g \in \Acc(q) \}$
        is either finite or co-finite.
        Consider a remainder $0 \le r \le g-1$ where $S_r$ is infinite.
        We need to show that $S_r$ is indeed co-finite.
        Let $i \in S_r$ with $i \ge |Q|$,
        i.e. there exists an accepting run $\pi$ from $q$ of length $r + i \cdot g$.
        Since $\pi$ has length at least $|Q|$ it must traverse a state $q$ in a non-transient SCC $C$.
        Choose $j_0$ such that $j_0 \cdot g \ge m(C)$
        where $m(C)$ is the reachability constant from Lemma~\ref{lemma-alon}.
        By Lemma~\ref{lemma-alon} for all $j \ge j_0$
        there exists a cycle from $q$ to $q$ of length $j \cdot g$.
        Therefore we can prolong $\pi$ to a longer accepting run by $j \cdot g$ symbols
        for any $j \ge j_0$.
        This proves that $x \in S_r$ for every $x \geq i + j_0$
        and that $S_r$ is co-finite.
\end{proof}

\begin{lem}
	\label{lem:ae}
	A set $X \subseteq \N$ is eventually $d$-periodic iff $X$ and $X+d$ are almost equal.
\end{lem}
\begin{proof}
Let $t \in \N$ be such that for all $x \ge t$ we have $x \in X$ if and only if $x + d \in X$. Then $X$ and $X + d$ are equal up to threshold $t+d$. Conversely, if $X =_t X+d$, then for all $x \ge t$ we have $x + d \in X$ if and only if $x + d \in X+d$, which is true if and only if $x \in X$.
\end{proof}

Two sets $X,Y \subseteq \N$ are {\em equal up to a threshold $t \in \N$}, in symbol $X =_t Y$, if for all $x \geq t$:
$x \in X$ iff $x \in Y$. Two sets $X,Y \subseteq \N$ are {\em almost equal} if they are equal up to some threshold $t \in \N$. 

\begin{restatable}{lemma-restatable}{almostequal}
\label{lem:con-run}
Let $C$ be a non-transient SCC in $B$, $p,q \in C$ and $s = \shift(p,q)$.
Then $\Acc(p)$ and $\Acc(q) + s$ are almost equal.
\end{restatable}
\begin{proof}
	Let $k \in \N$ such that $k \cdot g \ge m(C)$ where $m(C)$ is the large enough constant from Lemma~\ref{lemma-alon}.
	By Lemma~\ref{lemma-alon} there exists a run from $p$ to $q$ of length $s + k \cdot g$,
	and a run from $q$ to $p$ of length $(k+1) \cdot g - s$ (the latter number is congruent 
	to $\shift(q,p)$ modulo $g$).
	By prolonging accepting runs we obtain
	\[
		\Acc(q) + s + k \cdot g \subseteq \Acc(p) \mbox{ and } \Acc(p) + (k+1) \cdot g - s \subseteq \Acc(q).
	\]
	Adding $s + k \cdot g$ to both sides of the last inclusion yields
	\[
		\Acc(p) + (2k+1) \cdot g \subseteq \Acc(q) + s + k \cdot g \subseteq \Acc(p).
	\]
	By Lemmas~\ref{lem:acc-per} and~\ref{lem:ae} the three sets above are almost equal.
	Also $\Acc(q) + s + k \cdot g$ is almost equal to $\Acc(q) + s$
	by Lemmas~\ref{lem:acc-per} and~\ref{lem:ae}. Since almost equality is a transitive relation,
	this proves the statement.
\end{proof}

\begin{cor}
\label{cor:gt}
There exists a threshold $t \in \N$ such that
\begin{enumerate}
	\item $\Acc(q) =_t \Acc(q) + g$ for all $q \in Q$, and
	\item $\Acc(p) =_t \Acc(q) + \shift(p,q)$ for all non-transient SCCs $C$ and $p,q \in C$.
\end{enumerate}
\end{cor}

We fix the threshold $t$ from Corollary~\ref{cor:gt} for the rest of Section~\ref{sec-rand}.
The following lemma is the main tool to prove the correctness of our sliding window testers.
It states that if a word of length $n$ is accepted from $p$ and $\rho$ is any internal run from $p$ of length at most $n$, then, up to a bounded length prefix, $\rho$ can be extended to an accepting run of length $n$.
Formally, a run~$\pi$ {\em $k$-simulates} a run $\rho$
if one can factorize $\rho = \rho_1 \rho_2$ and $\pi = \pi' \rho_2$ where $|\rho_1| \le k$.

\begin{restatable}{lemma-restatable}{runsimulation}
\label{lem:sim}
If $\rho$ is an internal run starting from $p$ of length at most $n$ and $n \in \Acc(p)$,
then there exists an accepting run $\pi$ from $p$ of length $n$
which $t$-simulates $\rho$.
\end{restatable}
\begin{proof}
If $|\rho| \le t$, then we choose any accepting run $\pi$ from $p$ of length $n \in \Acc(p)$.
Otherwise, if $|\rho| > t$, then the SCC $C$ containing $p$ is non-transient
and we can factor $\rho = \rho_1\rho_2$ such that $|\rho_1| = t$
where $\rho_2$ leads from $p$ to $q$.
Set $s := \shift(q,p)$, which satisfies $s + |\rho_2| \equiv 0 \pmod g$
by the properties in Lemma~\ref{lemma-alon}.
Since $\Acc(q) =_t \Acc(p) + s$ by Corollary~\ref{cor:gt}, $n > t$ and $n \in \Acc(p)$,
we have $n + s \in \Acc(q)$.
Finally since $n + s \equiv n-|\rho_2| \pmod g$ and $n-|\rho_2| = n-|\rho|+t \ge t$
we know $n-|\rho_2| \in \Acc(q)$.
This yields an accepting run $\pi'$ from $q$ of length $n-|\rho_2|$.
Then $\rho$ is $t$-simulated by $\pi = \pi' \rho_2$.
\end{proof}

\subsection{Deterministic logspace tester}\label{sec:deterministic_ub}

\begin{proof}[Proof of Theorem~\ref{theorem:deterministic_ub}]
Let $n \in \N$ such that $n \ge |Q|$ (for $n < |Q|$ we use a trivial streaming algorithm which stores the window explicitly).
The algorithm maintains the set $\{ \ps(\pi_{w,q}) \mid q \in Q \}$ where $w \in \Sigma^n$ is the active window.
Initially this set is $\{ \ps(\pi_{w,q}) \mid q \in Q \}$ for $w = \Box^n$.
Now suppose $w = av$ for some $a \in \Sigma$ and the next symbol of the stream is $b \in \Sigma$, i.e. the new active window is $vb$.
For each transition $q \xleftarrow{b} p$ in $B$
we can compute $\ps(\pi_{vb,p})$ from $\ps(\pi_{av,q})$ as follows.
Suppose that $\ps(\pi_{av,q}) = (\ell_m,q_m) \cdots (\ell_1,q_1)$ where $q = q_1$.
\begin{itemize}
	\item If $p$ and $q$ belong to the same SCC, then we increment $\ell_1$ by one,
	else we append a new pair $(1,p)$.
	\item If $\ell_m > 0$ we decrement $\ell_m$ by one. If $\ell_m=0$ we remove the pair $(\ell_m,q_m)$
	and we decrement $\ell_{m-1}$ by one (in this case we must have $m > 1$ and $\ell_{m-1} > 0$).
\end{itemize}
The obtained path summary is $\ps(\pi_{vb,p})$.
This data structure can be stored with $\O(\log n)$
bits since it contains $|Q|$ path summaries, each of which can be stored in $\O(\log n)$ bits.

It remains to define a proper acceptance condition.
Consider the run $\pi = \pi_{w,q_0}$ such that $\pi_m \tau_{m-1} \pi_{m-1} \cdots \tau_1 \pi_1$ is the corresponding SCC-factorization
and $(\ell_m,q_m) \cdots (\ell_1,q_1)$ is the corresponding path summary.
The algorithm accepts if and only if $\ell_m = |\pi_m| \in \Acc(q_m)$.
If $w \in L$, then clearly $|\pi_m| \in \Acc(q_m)$.
If $|\pi_m| \in \Acc(q_m)$,
then the internal run $\pi_m$ can be $t$-simulated by an accepting run $\pi_m'$ of equal length by Lemma~\ref{lem:sim}.
The run $\pi_m' \tau_{m-1} \pi_{m-1} \cdots \tau_1 \pi_1$ is accepting
and witnesses that $\pdist(w,L) \le t$.
\end{proof}

\subsection{Randomized constant-space tester with two-sided error}
\label{sec:two-sided}
Let us first define a probabilistic counter. Consider a probabilistic data structure $Z$ representing a counter. 
Its operations are incrementing the counter (using random coins)
and querying whether the state of the counter is {\em low} or {\em high}.
Initially $Z$ is in a low state.
The random state reached after $k$ increments is denoted by $Z(k)$.
Given numbers $0 \le \ell < h$ (they will depend on our window size $n$) we say that $Z$ is an
{\em $(h,\ell)$-counter with error probability $\delta < \frac{1}{2}$}
if for all $k \in \N$ we have:
\begin{itemize}
	\item If $k \le \ell$, then $\Prob[Z(k) \text{ is high}] \le \delta$.
	\item If $k \ge h$, then $\Prob[Z(k) \text{ is low}] \le \delta$.
\end{itemize}
\begin{restatable}{lemma-restatable}{probabilisticcounter}
\label{lem:bernoulli}
For all $h,\ell,\eps > 0$ with $\ell \le (1-\eps) h + \O(1)$ there exists an $(h,\ell)$-counter $Z$ with error probability $1 / 3|Q|$ which internally stores $\O(\log(1/\eps))$ bits.
\end{restatable}
\begin{proof}
Since $\ell \le (1-\eps) h + \O(1)$, we can choose $\xi = \eps - \O(1)$ such that $\ell \le (1-\xi) h$.

We use the following probabilistic data structure from \cite{GanardiHL18}:
A {\em Bernoulli counter} $Z_p$ is parameterized by a probability $0 < p < 1$ and stores a single bit $x$. Initially we set $x = 0$, representing the low state. On every increment the bit $x$ is set to $1$ (representing the high state) with probability $p$, and is unchanged with probability $1-p$. After $i$ increments the bit has value $0$ with probability $(1-p)^i$, and value $1$ with probability $1 - (1-p)^i$.
Let us first show the following claim:

\begin{claim} \label{claim-bernoulli}
 For all $h,\ell,\xi > 0$ with $\xi < 1$ and $\ell \le (1-\xi) h$ there exists $0 < p < 1$ such that $Z_p$ is an $(h,\ell)$-counter with error probability $1/2 - \xi/8$. 
\end{claim}
\begin{claimproof}
We need to choose $p$ such that (i) $1-(1-p)^{(1-\xi) h} \le 1/2-\xi/8$, or equivalently, $1/2+\xi/8 \le (1-p)^{(1-\xi)h}$, and (ii) $(1-p)^h \le 1/2-\xi/8$, or  equivalently, $(1-p)^{(1-\xi)h} \le (1/2-\xi/8)^{1-\xi}$. It suffices to show
\begin{equation}
\label{eq:pow}
\frac{1}{2}+\frac{\xi}{8} \le \left( \frac{1}{2}-\frac{\xi}{8} \right)^{1-\xi},
\end{equation}
then one can pick $p = 1-(1/2-\xi/8)^{1/h}$. Note that (ii) holds automatically for this value of $p$.
Taking logarithms shows that \eqref{eq:pow} is equivalent to $\ln(4+\xi) - \ln 8 \le (1-\xi) \cdot (\ln(4-\xi) - \ln 8)$, and by rearranging we obtain $\ln(4+\xi) \le \ln(4-\xi) + \xi (\ln 8 - \ln(4-\xi))$. Since $\ln 8 - \ln(4-\xi) \ge \ln 8 - \ln 4 = \ln 2$, it suffices to prove
\begin{equation}
	\label{eq:log}
	\ln(4+\xi) \le \ln(4-\xi) + \xi \ln 2.
\end{equation}
One can verify $3\ln 2 \approx 2.0794 \ge 2$. We have:
\begin{align*}
	4+\xi &\le 4 + (3\ln 2 - 1) \xi = 4 + (4\ln 2 - 1) \xi - \xi \ln 2 \le \\
	&\le 4 + (4\ln 2 - 1) \xi - \xi^2 \ln 2 = (4-\xi)(\xi \ln 2 + 1)
\end{align*}
By taking logarithms and plugging in $\ln x \le x-1$ for all $x > 0$, we obtain
\[
	\ln(4+\xi) \le \ln(4-\xi) + \ln(\xi \ln 2 + 1) \le \ln(4-\xi) + \xi \ln 2
\]
This proves \eqref{eq:log} and hence \eqref{eq:pow}, and hence Claim~\ref{claim-bernoulli}.
\end{claimproof}
We now show the main claim of the lemma by probability amplification.
Let $Z$ be the counter which uses $m$ copies of $Z_p$ in parallel with independent random bits and returns the majority vote of the $m$ outputs. Notice that it suffices to store the sum of all bits, which takes $\O(\log m)$ bits of space.
	
Let us now estimate the error probability and choose $m$ suitably. Let $X_1, \dots, X_m$ be independent Bernoulli variables with 
$\Prob[X_i = 1] = 1/2 - \xi/8$. By Claim~\ref{claim-bernoulli}, $\Prob[X_i = 1]$ is an upper bound on the error probability of the $i$-th copy of $Z_p$.
Let $X = \sum_{i=1}^m X_i$. Then $\Prob[X \ge m/2]$ is an upper bound on the error probability of the probabilistic counter $Z$. 
We have $\mu = \mathbf{E}[X] = m(1/2 - \xi/8) =  \frac{m (4-\xi)}{8}$. Choosing $\delta = \frac{\xi}{4-\xi} \ge \frac{\xi}{4}$ we have $(1+\delta)\mu = m/2$ and $\mu \delta^2 = \xi m \delta/8 \ge \xi^2 m / 32$. The Chernoff bound \cite[Theorem~4.4]{MiUp17} states that
\[
\Prob[X \ge m/2] = \Prob\left[X \ge (1+\delta)\mu \right] \le \exp(-\mu \delta^2/3) \le \exp(-\xi^2 m/96).
\]
To enforce $\Prob[X \ge m/2] \le 1/(3|Q|)$ we choose $m = \left\lceil 96 \ln(3|Q|) / \xi^2 \right\rceil$. Hence the algorithm has space complexity $\O(\log m) = \O(\log (1/\xi) ) = \O(\log(1/\eps))$.
\end{proof}

Fix a parameter $0 < \eps < 1$ and a window length $n \in \N$.
Based on the previous concepts, we are now able to describe a randomized sliding window tester for a regular language $L$ with Hamming gap $\epsilon n$ that uses $\O(\log(1/\eps))$ bits.
Let $Z$ be the $(h,\ell)$-counter with error probability $1/(3|Q|)$ from Lemma~\ref{lem:bernoulli}
where $h = n - t$ and $\ell = (1-\eps)n + t + 1$. The counter is used to define so-called compact summaries of runs. 


\begin{defn}
\label{def:rep}
A {\em compact summary} $\text{cs} = (q_m, r_m, c_m) \cdots (q_2, r_2, c_2) (q_1, r_1, c_1)$ is a sequence of triples, where each triple $(q_i, r_i, c_i)$ consists of a state $q_i \in Q$, a remainder $0 \le r_i \le g-1$, and a state $c_i$ of the $(h,\ell)$-counter $Z$. The state $c_1$ is always set to low, and $r_1 = 0$.

A compact summary $(q_m, r_m, c_m) \cdots (q_1, r_1, c_1)$ {\em represents} a run $\pi$ if
the SCC-factorization of $\pi$ has the form $\pi_m \tau_{m-1} \pi_{m-1} \cdots \tau_1 \pi_1$, and the following properties hold:
\begin{enumerate}
\item for all $1 \le i \le m$, $\pi_i$ starts in $q_i$; 
\item for all $2 \le i \le m$, if $|\tau_{i-1} \pi_{i-1} \cdots \tau_1 \pi_1| \le (1-\eps)n + t + 1$, then $c_i$ is a low state; and if $|\tau_{i-1} \pi_{i-1} \cdots \tau_1 \pi_1| \ge n-t$, then $c_i$ is a high state;
\item for all $2 \le i \le m$, $r_i = |\tau_{i-1} \pi_{i-1} \cdots \tau_1 \pi_1| \pmod {g}$.
\end{enumerate}
\end{defn}

The idea of a compact summary is visualized in Figure~\ref{fig:compact}. If $m > |Q|$ then the above compact summary cannot represent a run.
Therefore, we can assume that $m \leq |Q|$.
For every triple $(q_i, r_i, c_i)$, the entries $q_i$ and $r_i$ only depend on the rDFA $B$, and hence can be stored with $\O(1)$ bits. Every state $c_i$ of the probabilistic counter needs $\O(\log(1/\eps))$ bits. Hence, a compact summary can be stored in $\O(\log(1/\eps))$ bits. In contrast to Theorem~\ref{theorem:deterministic_ub}, we maintain a set of compact summaries which represent all runs of $B$ on the {\em complete} stream read so far (not only on the active window) with high probability.

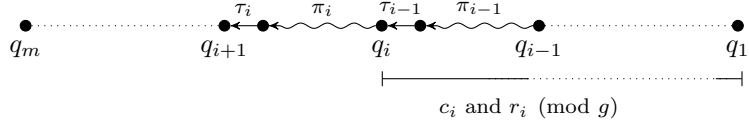
\begin{figure}[t]
\centering
\begin{tikzpicture}
  \tikzset{every edge/.style={draw,->,>=stealth'}}
  \tikzstyle{small} = [circle,draw=black,fill=black,inner sep=.5mm]
  
  \node[small,label=below:$q_{i-1}$] (q1) {};
  \node[small, left = 4em of q1] (p1) {};
  \node[small, label=below:$q_i$,left = 1em of p1] (q2) {};
  \node[small, left = 4em of q2] (p2) {};
  \node[small, label=below:$q_{i+1}$,left = 1em of p2] (q3) {};

  \node[small, right = 7em of q1, label=below:$q_1$] (q0) {};
  \node[small, left = 7em of q3, label=below:$q_m$] (qm) {};

  \draw (q1) edge[decorate,decoration={snake,amplitude=1pt}]  node[above] {\footnotesize $\pi_{i-1}$} (p1);
  \draw (p1) edge node[above] {\footnotesize $\tau_{i-1}$} (q2);
  \draw (q2) edge[decorate,decoration={snake,amplitude=1pt}]  node[above] {\footnotesize $\pi_i$} (p2);
  \draw (p2) edge node[above] {\footnotesize $\tau_i$} (q3);
  
  \draw (q1) edge[-,dotted] (q0);
  \draw (q3) edge[-,dotted] (qm);
 
  \draw[|-] (q2) ++(0,-2em) -- ++(5.5em,0) node[label=below:{\footnotesize $c_i$ and $r_i\pmod {g}$}] {};
  \draw[dotted] (q2) ++(4em,-2em) -- ++(9em,0) node {};
  \draw[-|] (q2) ++(12.5em,-2em) -- ++(1em,0);
\end{tikzpicture}
\caption{A compact summary of a path $\pi$.}\label{fig:compact}
\end{figure}

\begin{restatable}{proposition-restatable}{compactsummaries}
\label{prop:ps-alg}
For a given input stream $w \in \Sigma^*$, we can maintain a set of compact summaries $S$ containing for each $q \in Q$ a compact summary $\text{cs}_q \in S$ starting in $q$ such that $\text{cs}_q$ represents the unique run $\pi_{w,q}$ 
with probability at least $2/3$.
\end{restatable}
\begin{proof}
For each state in $Q$, we initialize the compact summary so that it represents the run $\pi_{\lambda,q}$ (recall that $\lambda$ is the empty word). Consider a compact summary $\text{cs} = ( q_m, r_m, c_m ) \cdots (q_1, r_1, c_1)$, which represents a run $\pi_{x,q_1}$. We prolong $\text{cs}$ by a transition $q_1 \xleftarrow{a} p$ in $B$ as follows:
\begin{itemize}
\item if $p$ and $q$ are not in the same SCC, then we increment all counter states $c_i$, increment all remainders $r_i$ mod $g$, and append a new triple $(p, 0, c_1)$;
\item if $p$ and $q$ belong to the same SCC, then we increment all counter states $c_i$ for $2 \le i \le m$, increment the remainder $r_i$ mod $g$ for $2 \le i \le m$, and replace $q_1$ by $p$.
\end{itemize}
If $a \in \Sigma$ is the next input symbol of the stream, then $S$ is updated to the new set $S'$ of compact summaries by iterating over all transition $q \xleftarrow{a} p$ in $B$ and prolonging the compact summary starting in $q$ by the transition.

To verify correctness, consider $\text{cs} = (q_m, r_m, c_m) \cdots (q_1, r_1, c_1)$ as a compact summary computed by the algorithm. Properties (1) and (3) from Definition~\ref{def:rep} are satisfied by construction. Furthermore, since $m \le |Q|$ the probability that Property (2) or (4) is violated is at most $m/(3|Q|) \le 1/3$ by the union bound.
\end{proof}

It remains to define an acceptance condition on compact summaries.
For every $q \in Q$ we define $\Acc_{\mathit{mod}}(q) = \{ \ell \pmod {g} : \ell \in \Acc(q) \text{ and } \ell \ge t\}$, which is intuitively speaking the set of accepting remainders.
Consider a compact summary $\text{cs} = (q_m, r_m, c_m) \cdots (q_1, r_1, c_1)$. Since $c_1$ is the low initial state of the probabilistic counter, there exists a maximal index $i \in \{1, \dots, m\}$ such that $c_i$ is low.
We say that $\text{cs}$ is {\em accepting} if $n - r_i \pmod g \in \Acc_{\mathit{mod}}(q_i)$.

\begin{restatable}{proposition-restatable}{onesideaccept}
\label{prop:correctness}
Assume that $\eps n \ge t$. Let $w \in \Sigma^*$ with $|w| \ge n$ and let $\text{cs}$ be a compact summary which represents $\pi_{w,q_0}$.
\begin{enumerate}
\item If $\last_n(w) \in L$, then $\text{cs}$ is accepting.
\item If $\text{cs}$ is accepting, then $\pdist(\last_n(w),L) \le \eps n$.
\end{enumerate}
\end{restatable}
\begin{proof}
Consider the SCC-factorization of $\pi = \pi_{w,q_0} = \pi_m \tau_{m-1} \cdots \tau_1 \pi_1$. Let $\text{cs} = (q_m, c_m, r_m) \cdots (q_1, c_1, r_1)$ be a compact summary representing $\pi$. Thus, $q_1 = q_0$.
Consider the maximal index $1 \le i \le m$ where $c_i$ is low,
which means that $|\tau_{i-1} \pi_{i-1} \cdots \tau_1 \pi_1| < n-t$ by Definition~\ref{def:rep}(4).
The run of $B$ on $\last_n(w)$ has the form $\pi_k' \tau_{k-1} \pi_{k-1} \cdots \tau_1 \pi_1$
for some suffix $\pi_k'$ of $\pi_k$.
We have $|\pi_k' \tau_{k-1} \cdots \pi_i| = n - |\tau_{i-1} \pi_{i-1} \cdots \tau_1 \pi_1| > t$.
By Definition~\ref{def:rep}(2) we know that
\[
	r_i = |\tau_{i-1} \pi_{i-1} \cdots \tau_1 \pi_1| \pmod {g} = n - |\pi_k' \tau_{k-1} \cdots \pi_i| \pmod {g}.
\]
For point 1~assume that $\last_n(w) \in L$. Thus, $\pi_k' \tau_{k-1} \pi_{k-1} \cdots \tau_1 \pi_1$
is an accepting run starting in $q_0$.
By Definition~\ref{def:rep}(1), the run 
$\pi_k' \tau_{k-1} \cdots \pi_i$ starts in $q_i$.
Hence, $\pi_k' \tau_{k-1} \cdots \pi_i$ is an accepting run from $q_i$ of length at least $t$.
By definition of $\Acc_{\mathit{mod}}(q_i)$ we have
$|\pi_k' \tau_{k-1} \cdots \pi_i| \pmod {g} = n - r_i \pmod {g} \in \Acc_{\mathit{mod}}(q_i)$,
and therefore $\text{cs}$ is accepting.

For point 2~assume that $\text{cs}$ is accepting,
i.e.
\[
	n - r_i \pmod g = |\pi_k' \tau_{k-1} \cdots \pi_i|  \pmod g \in \Acc_{\mathit{mod}}(q_i).
\]
Recall that $|\pi_k' \tau_{k-1} \cdots \pi_i| > t$.
By definition of $\Acc_{\mathit{mod}}(q_i)$ there exists an accepting run from $q_i$
whose length is congruent to $|\pi_k' \tau_{k-1} \cdots \pi_i|$ mod $g$
and at least $t$.
By Corollary~\ref{cor:gt}(1) we derive that $|\pi_k' \tau_{k-1} \cdots \pi_i| \in \Acc(q_i)$.
We claim that $|\pi_i \tau_{i-1} \pi_{i-1} \cdots \tau_1 \pi_1| \ge (1-\eps)n+t$ by a case distinction.
If $i = m$, then clearly $|\pi_i \tau_{i-1} \pi_{i-1} \cdots \tau_1 \pi_1| \ge n \ge (1-\eps)n+t$.
If $i < m$, then $c_{i+1}$ is high by maximality of $i$,
which implies $|\tau_i \pi_i \cdots \tau_1 \pi_1| > (1-\eps)n + t+1$ by Definition~\ref{def:rep}(3).
Since $\tau_i$ has length one, we have $|\pi_i \tau_{i-1} \pi_{i-1} \cdots \tau_1 \pi_1| > (1-\eps)n+t$.

Since $|\pi_k' \tau_{k-1} \cdots \pi_i| \in \Acc(q_i)$,
we can apply Lemma~\ref{lem:sim} and obtain an accepting run $\rho$
of length $|\pi_k' \tau_{k-1} \cdots \pi_i| \in \Acc(q_i)$ starting in $q_i$
which $t$-simulates the internal run $\pi_i$.
The prefix distance from $\rho$ to $\pi_k' \tau_{k-1} \cdots \pi_i$ is at most
\[
	|\pi_k' \tau_{k-1} \cdots \tau_i| + t = n - |\pi_i \tau_{i-1} \pi_{i-1} \cdots \tau_1 \pi_1| + t \le n - (1-\eps)n = \eps n.
\]
Therefore the accepting run $\rho \tau_{i-1} \pi_{i-1} \cdots \tau_1 \pi_1$
and $\pi_k' \tau_{k-1} \pi_{k-1} \cdots \tau_1 \pi_1$ have prefix distance at most $\eps n$ as well.
This implies  $\pdist(\last_n(w),L) \le \eps n$.
\end{proof}



\begin{proof}[Proof of Theorem~\ref{theorem:two-sided}]
Assume that $\eps n \ge t$, otherwise we use a trivial streaming algorithm that stores the window explicitly with $\O(1/\eps)$ bits.
We use the algorithm from Proposition~\ref{prop:ps-alg} for each incoming symbol from the stream. To initialize, we run the algorithm on $\square^n$.
The algorithm accepts if the computed compact summary starting in $q_0$ is accepting.
From Proposition~\ref{prop:ps-alg} and~\ref{prop:correctness} we get:
\begin{itemize}
\item If $\pdist(\last_n(w),L) > \eps n$, then the algorithm rejects with probability at least $2/3$.
\item If $\last_n(w) \in L$, then the algorithm accepts with probability at least $2/3$.
\end{itemize}
This concludes the proof of the theorem.
\end{proof}
Comparing Theorems~\ref{theorem:deterministic_ub} and \ref{theorem:two-sided} leads to the question whether one can replace the Hamming gap $\gamma(n) = \eps n$ in Theorem~\ref{theorem:two-sided} by $\gamma(n) = o(n)$ while retaining constant space at the same time. We show that this is not the case:

\begin{restatable}{lemma-restatable}{needlineargap}
Let $L = a^* \subseteq \{a,b\}^*$. Every randomized sliding window tester with two-sided error for $L$ with Hamming gap $\gamma(n)$ needs space $\Omega(\log n - \log \gamma(n))$ for infinitely many $n$.
\end{restatable}
\begin{proof}
We prove the lemma by a reduction from the randomized one-way communication complexity of the greater-than-function.\footnote{A similar reduction was used
in \cite{GanardiHL18}.} The setting is the following: Alice (resp.~Bob) holds a number $i \in \{1,\ldots,m\}$ (resp., $j \in \{1,\ldots,m\}$). Moreover, both parties receive a random string. Then Alice sends a message to Bob (depending on her input $i$ and her random string), and Bob has to decide whether $i > j$ or $i \le j$ holds. It is known that in every such one-way protocol, where Bob gives a correct answer with probability at least $2/3$, Alice has to send $\Omega(\log m)$ bits to Bob \cite[Theorem~3.8]{KremerNR99}.

Consider a 
randomized sliding window tester for $a^*$ with Hamming gap $\gamma(n)$ 
that uses space $s(n)$.
Fix a window size $n$, which is divisible by $k := \gamma(n)+1$. Let $m = n/k$.
We divide the window into $m$ 
blocks of length $k$. We then obtain a randomized one-way protocol for the greater-than-function on the interval $\{1,\ldots,m\}$: 
Alice produces from her input $i \in \{1,\ldots,m\}$ the word 
$w_i = a^{(i-1)k} b^{k} a^{(m-i)k}$.
She then runs the randomized sliding window tester on $w_i$ (using her random bits) 
and sends the final memory content ($s(n)$ bits)
to Bob. Bob continues the run of the randomized sliding window tester (starting from the transferred memory
content) with the input stream $a^{j k}$.
He obtains the memory content reached after the input 
$a^{(i-1)k} b^{k} a^{(m-i+j)k}$. Finally, Bob outputs the answer given by 
the randomized sliding window tester. If $i \le j$, then the window content at the end is $a^n$ and hence
belongs to $a^*$. On the other hand, if $i > j$, then the window content at the end contains the block 
$b^{k}$, hence, the Hamming distance between the window content and $a^*$ is at least
$\gamma(n)+1$. This implies that Bob will give a correct answer with probability at least $2/3$. 
It follows that $s(n) \in \Omega(\log m) = \Omega(\log n - \log \gamma(n))$. Note that for the case $\gamma(n) \leq n^{\epsilon}$ for a constant $\epsilon > 0$ we obtain
$s(n) \in \Omega(\log n)$.
\end{proof}

\subsection{Randomized loglogspace tester with one-sided error}\label{sec:one-sided_ub}
Let $L$ be a finite union of trivial regular languages and suffix-free regular languages. In this section, we present a randomized sliding window tester for $L$ with one-sided error and Hamming gap $\gamma(n) = \eps n$ that uses space $\O(\log \log n)$.
By Remark~\ref{rem-finite-union}  and Theorem~\ref{prop-trivial1}, it suffices to consider the case when $L$ is a suffix-free regular language. 
As in Section~\ref{sec-ps} we fix an rDFA $B= (Q,\Sigma,F,\delta,q_0)$ for $L$ such that $g(C)=g$ for all SCCs of $A$. 
Since $L$ is suffix-free, $B$ has the property that no final state can be reached from a final state by a non-empty run. We decompose $B$ into a finite union of {\em partial automata}, similar to \cite{GHKLM18}. 

\begin{defn}\label{def:path-descript}

We call a sequence
$$(q_k,a_k,p_{k-1}),C_{k-1},\dots ,(q_2,a_2,p_1),C_1,(q_1,a_1,p_0),C_0,q_0$$
a {\em path description}
if $C_{k-1}, \dots, C_0$ is a chain (read from right to left) in the SCC-ordering of $B$,
$p_i, q_i \in C_i$, $q_{i+1} \xleftarrow{a_{i+1}} p_i$ is a transition in $B$
for all $0 \le i \le k-1$, and $q_k \in F$.
\end{defn}

Each path description defines a {\em partial rDFA} $B_P = (Q_P,\Sigma,\{q_k\},\delta_P,q_0)$ by restricting $B$ to the state set $Q_P = \bigcup_{i=0}^{k-1}C_i \cup \{q_k\}$, restricting the transitions of $B$ to internal transitions from the SCCs $C_i$ and the transitions $q_{i+1} \xleftarrow{a_{i+1}} p_i$, and declaring $q_k$ to be the only final state. The rDFA is partial since for every state $p_i$ and every symbol $a \in \Sigma$ there exists at most one transition $q \xleftarrow{a} p_i$. Since the number of path descriptions $P$ is finite and $L(B) = \bigcup_P L(B_P)$, it suffices to provide a sliding window tester for $L(B_P)$ (we again use Remark~\ref{rem-finite-union} here).

From now on, we fix a path description $P$ from Definition~\ref{def:path-descript} and the partial automaton $B_P = (Q_P,\Sigma,\{q_k\},\delta_P,q_0)$ corresponding to it. 
The acceptance sets $\Acc(q)$ are defined with respect to $B_P$.
 If all $C_i$ are transient, then $L(B_P)$ is a singleton and we can use a trivial sliding window tester with space complexity $\O(1)$. Now assume the contrary and let $0 \le e \le k-1$ be maximal such that $C_e$ is nontransient. 

\begin{restatable}{lemma-restatable}{acceptinglengthsonesided} \label{lemma-si-ri}
There exist numbers $r_0, \ldots, r_{k-1}, s_0, \ldots, s_e \in \N$ such that the following hold:
\begin{itemize}
\item For all $e+1 \le i \le k$, the set $\Acc(q_i)$ is a singleton.
\item For all $0 \le i \le e$, $\Acc(q_i) =_{s_i} \sum_{j=i}^{k-1} r_j + g \N$.
\item Every run $\pi$ from $q_i$ to $q_{i+1}$ $(0 \le i \le k-1)$ satisfies $|\pi| \equiv r_i \pmod g$.
\end{itemize}
\end{restatable}
\begin{proof}
The first statement of the lemma follows immediately from the definition of transient SCCs. 

Let us now show the second and third statement of the lemma.
Let $0 \le i \le k-1$ and let $N_i$ be the set of lengths of runs of the form $q_{i+1} \xleftarrow{a_{i+1}} p_i \xleftarrow{w} q_i$ in $B_P$. If $C_i$ is transient, then $N_i = \{1\}$. Otherwise, by Lemma~\ref{lemma-alon} there exist a number $r_i \in \N$ and a cofinite set $D_i \subseteq \N$ such that $N_i = r_i + gD_i$. We can summarize both cases by saying that there exist a number $r_i \in \N$ and a set $D_i \subseteq \N$ which is either cofinite or $D_i = \{0\}$ such that $N_i = r_i + gD_i$. This yields the third statement. Moreover,
the acceptance sets in $B_P$ satisfy
\[
	\Acc(q_i) = \sum_{j=i}^{k-1} N_j = \sum_{j=i}^{k-1} (r_j + gD_j) = \sum_{j=i}^{k-1} r_j + g \sum_{j=i}^{k-1} D_j.
\]
For all $0 \le i \le e$ we get $\Acc(q_i) =_{s_i} \sum_{j=i}^{k-1} r_j + g \N$ for some threshold $s_i \in \N$ (note that a non-empty sum of cofinite subsets of $\N$ is again cofinite).
\end{proof}
Let us fix the numbers $r_i$ and $s_i$ from Lemma~\ref{lemma-si-ri}. 
Let $p$ be a random prime with $\Theta(\log \log n)$ bits. Define a threshold 
\[ s = \max \bigg\{k, \sum_{j=0}^{k-1} r_j, s_0, \dots, s_e \bigg\}\] 
and for a word $w \in \Sigma^*$ define the function
$\ell_w \colon Q \to \N \cup \{\infty\}$ where
\[
	\ell_w(q) = \inf \{ \ell \in \N \mid \delta_P(\last_\ell(w),q) = q_k  \}
\]
(we set $\inf \emptyset = \infty$). We now define an acceptance condition on $\ell_w(q)$. If $n \notin \Acc(q_0)$, we always reject. Otherwise, we accept $w$ iff $\ell_w(q_0) \equiv n$
modulo our randomly chosen prime $p$.

\begin{restatable}{lemma-restatable}{acceptingconditiononesided}
\label{lm:acceptingconditiononesided}
Let $n \in \Acc(q_0)$ be a window size with $n \ge s + |Q_P|$ and $w \in \Sigma^*$ with $|w| \ge n$. There exists a constant $c > 0$ such that:
\begin{enumerate}
\item if $\last_n(w) \in L(B_P)$, then $w$ is accepted with probability $1$;
\item if $\pdist(\last_n(w), L(B_P)) > c$, then $w$ is rejected with probability at least $2/3$.
\end{enumerate}
\end{restatable}
\begin{proof}
Assume first that $\last_n(w) \in L(B_P)$. Since $L(B_P) \subseteq L$ is suffix-free, $\ell_w (q_0) = n \pmod p$ and $w$ is accepted with probability $1$.

Consider now the case when $\last_n(w) \notin L(B_P)$. By definition, in this case $\ell_w (q_0) \neq n$. In other words, only two cases are possible: either $\ell_w(q_0) < n$, or $\ell_w(q_0) > n$.  If $\ell_w(q_0) < n$, then by the choice of $p$  $\ell_w(q_0) \not \equiv n \pmod p$ with probability at least $2/3$. 

We finally consider the case $\ell_w(q_0) > n$. We will show that in this case the prefix distance between $\last_n(w)$ and $L(B_P)$ is bounded by a constant $c$, which means that we can either accept or reject. Let $\pi$ be the run of $B_P$ on $\last_n(w)$ starting from the initial state $q_0$, and let $\pi = \pi_m \tau_{m-1} \pi_{m-1} \cdots \tau_0 \pi_0$ be its SCC-factorization. 
We have $|\pi|=n$.
Since $\ell_w(q_0) > n$, the run $\pi$ can be strictly prolonged to a run to $q_k$ and hence we must have $m < k$. 
For all $0 \leq i \leq m$, the run $\pi_i$ is an internal run in the SCC $C_i$ from $q_i$ to $p_i$.
For all $0 \le i \le m-1$ we have $\tau_i = (q_{i+1} \xleftarrow{a_{i+1}} p_i)$  and $|\tau_i \pi_i| \equiv r_i \pmod g$, where the latter follows from the third
statement in Lemma~\ref{lemma-si-ri}.
We claim that there exists an index $0 \le i_0 \le m$ such that the following three properties hold:
\begin{enumerate}
\item \label{it:i} $q_{i_0}$ is nontransient,
\item \label{it:ii} $|\pi_m \tau_{m-1} \pi_{m-1} \cdots \tau_{i_0} \pi_{i_0}| \ge s$,
\item \label{it:iii} $|\pi_m \tau_{m-1} \pi_{m-1} \cdots \tau_{{i_0}+1} \pi_{{i_0}+1}| \le s + |Q_P|$. 
\end{enumerate}
Indeed, let $0 \le i \le m$ be the smallest integer such that $q_i$ is nontransient (recall that $n \ge |Q_P|$ and hence $\pi$ must traverse a nontransient SCC). Then $\tau_{i-1} \pi_{i-1} \cdots \tau_0 \pi_0$ only passes transient states and hence its length is bounded by $|Q_P|$. Therefore, 
\begin{eqnarray*}
|\pi_m \tau_{m-1} \pi_{m-1} \cdots \tau_i \pi_i| & = & n - |\tau_{i-1} \pi_{i-1} \cdots \tau_0 \pi_0| \\
& \ge & n-|Q_P| \ge s
\end{eqnarray*}
Now let $0 \le {i_0} \le m$ be the largest integer satisfying Properties~\ref{it:i} and~\ref{it:ii}. If $\pi_m \tau_{m-1} \pi_{m-1} \cdots \tau_{{i_0}+1} \pi_{{i_0}+1}$ only passes transient states, then its length is bounded by $m-{i_0} \le s + m$, and we are done. Otherwise, let ${i_0}+1 \le j \le m$ be the smallest integer such that $q_j$ is nontransient. The run $\tau_{j-1} \pi_{j-1} \cdots \tau_{{i_0}+1} \pi_{{i_0}+1}$ only passes transient states and therefore it has length $j-{i_0}-1$.
By maximality of ${i_0}$, we have $|\pi_m  \tau_{m-1} \pi_{m-1} \cdots \tau_j \pi_j| < s$
and hence Property~\ref{it:iii} holds:

\begin{align*}
	|\pi_m \tau_{m-1} \pi_{m-1} \cdots \tau_{{i_0}+1} \pi_{{i_0}+1}| &= |\pi_m \cdots \tau_j \pi_j| + |\tau_{j-1} \pi_{j-1} \cdots \tau_{{i_0}+1} \pi_{{i_0}+1}|\\
	&< s + j - {i_0} \\
	&\le s + m.
\end{align*}
Let $0 \le {i_0} \le m$ be the index satisfying Properties~\ref{it:i}-\ref{it:iii}. Since $q_{i_0}$ is nontransient, we have ${i_0} \le e$
and therefore $\Acc(q_{i_0}) =_s \sum_{j={i_0}}^{k-1} r_j + g \N$ by the second
statement in Lemma~\ref{lemma-si-ri}. We have $|\pi_m \tau_{m-1} \pi_{m-1} \cdots \tau_{i_0} \pi_{i_0}| \in \Acc(q_{i_0})$ because it is larger than $s$ (by Property~\ref{it:ii}) and
\begin{align*}
	|\pi_m \tau_{m-1} \pi_{m-1} \cdots \tau_{i_0} \pi_{i_0}|  &= n - |\tau_{{i_0}-1} \pi_{{i_0}-1} \cdots \tau_0 \pi_0|\\
	& \equiv n - \sum_{j=0}^{{i_0}-1} r_j \pmod g\\
	&\equiv \sum_{j={i_0}}^{k-1} r_j \pmod g
\end{align*}
where the last congruence follows from $n \in \Acc(q_0) =_{s} \sum_{j=0}^{k-1} r_j + g \N$.
By Lemma~\ref{lem:sim} there exists an accepting run $\pi'$ of length
$|\pi_m \tau_{m-1} \pi_{m-1} \cdots \tau_{i_0} \pi_{i_0}|$ which $t$-simulates $\pi_{i_0}$.
The prefix distance between $\pi' \tau_{i-1} \pi_{{i_0}-1} \cdots \tau_0 \pi_0$ 
and $\pi = \pi_m \tau_{m-1} \pi_{m-1} \cdots \tau_0 \pi_0$ is at most

\begin{align*}
	|\pi_m \tau_{m-1} \pi_{m-1} \cdots \tau_{i_0}| + t &= |\pi_m \tau_{m-1} \pi_{m-1} \cdots \tau_{{i_0}+1} \pi_{{i_0}+1}| + 1 + t\\
	&\le 1 + s + m + t
\end{align*}
by Property~\ref{it:iii}.
\end{proof}

\begin{proof}[Proof of Theorem~\ref{theorem:one-sided_ub}]
Let $n \in \N$ be the window size. From the discussion above, it suffices to show a tester for a fixed partial automaton $B_P$. Assume $n \ge s + |Q|$, otherwise a trivial tester can be used. If $n \notin \Acc(q_0)$, the tester always rejects. Otherwise, the tester picks a random prime $p$ with $\Theta(\log \log n)$ bits and maintains $\ell_w(q) \pmod p$ for all $q \in Q_P$, where $w$ is the stream read so far, which requires $\O(\log \log n)$ bits. When a symbol $a \in \Sigma$ is read, we can update $\ell_{wa}$ using $\ell_w$: 
If $q = q_k$, then $\ell_{wa}(q) = 0$, otherwise $\ell_{wa}(q) = 1 + \ell_w(\delta_P(a,q)) \pmod p$
where $1 + \infty = \infty$. The tester accepts if $\ell_w(q_0) \equiv n \pmod p$. 
Lemma~\ref{lm:acceptingconditiononesided} guarantees correctness of the tester in the one-sided error setting.
\end{proof}

\section{Lower bounds}\label{section-lower-bounds}


A sliding window algorithm can be naturally seen as a family of finite automata (see~\cite{GHKLM18,GanardiHL18}).
We make use of this viewpoint in order to prove the lower bounds of Theorem~\ref{theorem:deterministic_lb} and Theorem~\ref{theorem:one-sided_lb}.
To get the strongest possible statements, we prove those lower bounds for so-called nondeterministic and co-nondeterministic sliding window testers.

A {\em nondeterministic finite automaton} (NFA) is a tuple $A = (Q,\Sigma,I,\delta,F)$
consisting of a finite set of states $Q$, a finite alphabet $\Sigma$,
a set of initial states $I \subseteq Q$, a transition relation $\delta \subseteq Q \times \Sigma \times Q$ and a set of final states $F \subseteq Q$.
Runs in NFAs are  defined similarly to DFAs and rDFAs. Formally, a run in the NFA $A$ is a sequence 
$(q_0, a_1, q_1, a_2, q_2, \ldots, a_n, q_n)$ such that $(q_{i-1}, a_i, q_i) \in \delta$ for all $1 \leq i \leq n$.
A word $w$ is accepted by $A$ ($w \in L(A)$ for short) if it labels a run from an initial state to a final state.

\begin{defn} \label{def-nondet-test}
A {\em nondeterministic sliding window tester} $\mathcal{A}=(A_n)_{n \geq 0}$ for the language $L$ with Hamming gap $\gamma(n)$ is a family of NFAs $A_n$ 
such that for each window size $n\ge 0$ and each stream $w\in\Sigma^*$ the following holds:
\begin{enumerate}
\item if $\last_n(w) \in L$, then $w \in L(A_n)$; \label{nfa1}
\item if $\dist(\last_n(w),L) > \gamma(n)$, then $w \notin L(A_n)$. \label{nfa2}
\end{enumerate}
\end{defn}
One can view every $A_n$ as a nondeterministic streaming algorithm that updates its memory state nondeterministically depending on the current input symbol.
Note that in order to have $\last_n(w) \in L$, it is enough to have at least one run of $A_n$ on $w\in\Sigma^*$ from an initial state to an accepting state.
This is equivalent to require that the active window is accepted by the algorithm with some probability greater than $0$ (if we assign to every state $q$ and every
symbol $a$ a probability distribution on the outgoing $a$-transitions of $q$).
On the other hand, if $\dist(\last_n(w),L) > \gamma(n)$, then all runs of $A_n$ on $w\in\Sigma^*$ from an initial state end in non-accepting states, i.e. the active window is rejected with probability $1$.

A second concept we use in this section are {\em coNFA}s. The only difference to NFAs is that a word $w$ is accepted by a coNFA $A$ if all runs on $w$ that begin in an initial state have to end in an accepting state. In other words, a word $w$ is rejected by $A$ if and only if there is at least one run on $w$ from an initial state to a non-accepting state.
A co-nondeterministic sliding window tester $\mathcal{A}=(A_n)_{n \geq 0}$ for $L$ with Hamming gap $\gamma(n)$ 
is a family of coNFAs $A_n$ such that for each window size $n\ge 0$ and each stream $w\in\Sigma^*$ the properties \ref{nfa1} and \ref{nfa2} in Definition~\ref{def-nondet-test} hold.
So if $\last_n(w) \in L$, then all runs of $A_n$ on $w\in\Sigma^*$ that start in an initial state end in an accepting state.
In other words, the algorithm accepts with probability $1$.
If $\dist(\last_n(w),L) > \gamma(n)$, then there is at least one run of $A_n$ on $w\in\Sigma^*$ that starts in an initial state and ends in a non-accepting state, i.e. the algorithm rejects with probability strictly greater than $0$.

Let $\mathcal{A}=(A_n)_{n \geq 0}$ be a (co-)nondeterministic sliding window tester and let $Q_n$ be the state set of $A_n$. Then the {\em space consumption} of $\mathcal{A}$ is defined as $s_{\mathcal{A}}(n)=\lceil \log |Q_n| \rceil$. This reflects the fact that 
states from $Q_n$ can be encoded with $s_{\mathcal{A}}(n)$ many bits.

We can now state our general lower bounds. 

\subsection{Nondeterministic lower bound}\label{sec:deterministic_lb}

\begin{restatable}{theorem-restatable}{lbnondet}
 \label{theorem-lower-bound}
Let $L$ be regular and nontrivial.
Then there is a constant $\eps_0$, $0 < \eps_0 \le 1$,
such that for every $0 \le \eps < \eps_0$,
every nondeterministic sliding window tester for $L$ with Hamming gap $\eps n$
uses space at least $\log_2 n + \log_2 (1-\eps/\eps_0) - \O(1)$ on an infinite set of window sizes $n$ (that only depends on $L$).
\end{restatable}
\begin{proof}
By Lemma~\ref{cuts}, $\cut_{i,j}(L)$ is not a length language for all $i, j \ge 0$.
Let $N$ be the set of lengths from Proposition~\ref{prop-nontriv}
such that $L|_N$ is infinite and excludes some factor $w_f$. Let $c = |w_f| >0$ and $\eps_0 = 1/c$.
Since $N$ is an arithmetic progression, $L|_N$ is regular. 
Recall that every word $v$ that contains $k$ disjoint occurrences of $w_f$ has Hamming distance
at least $k$ from any word in $L|_N$. 
Let $A = (Q,\Sigma,q_0,\delta,F)$ be a DFA for $L|_N$.
Since $L(A)$ is infinite, there must exist words $x,y,z$ such that $y \neq \lambda$ and 
for $\delta(q_0,x) = q$ we have $\delta(q, y) = q$ and $\delta(q,z) \in F$. Let $d = |xz|$ and $e = |y|>0$.

Consider a nondeterministic sliding window tester $\mathcal{A} = (A_n)_{n \geq 0}$ for $L$ with Hamming gap $\eps n$
for some $\eps < \eps_0$.
Fix a window length $n\in N$ and 
define for $k \geq 0$ the input streams $u_k = w_f^n x y^k$ and $v_k = u_k z = w_f^n x y^k z$.
 Let $\alpha = c \eps < 1$.
If $0 \le k \le \lfloor \frac{(1-\alpha) n - c - d}{e} \rfloor$, then the suffix of $v_k$ of length $n$ 
contains at least
$$
\bigg\lfloor\frac{n - d - e k}{c} \bigg\rfloor \geq 
\bigg\lfloor\frac{n - d - (1-\alpha) n + c + d}{c} \bigg\rfloor = 
\bigg\lfloor\frac{\alpha n + c}{c} \bigg\rfloor = \lfloor \eps n + 1 \rfloor > \eps n
$$
many disjoint occurrences of $w_f$.
Hence, after reading any of the input streams $v_k$ for 
$0 \le k \le \lfloor \frac{(1-\alpha) n - c - d}{e} \rfloor$, the NFA $A_n$ has to reject with 
probability one, i.e., every run of $A_n$ on $v_k$ that starts in an initial state has to end in a rejecting state.

Assume now that the window size $n$ satisfies $n \geq d$ and $n \equiv d \pmod e$. Write $n = d + l e$ for some $l \geq 0$.
Note that each $n$ with this property satisfies $n\in N$ since $xy^lz\in L|_N$.
We have $l > \lfloor \frac{(1-\alpha) n - c - d}{e} \rfloor$.
The suffix of $v_l = w_f^n x y^l z$ of length $n$ is $x y^l z \in L|_N$. Therefore $A_n$ accepts $v_l$, i.e.,
there exists a run $\pi$ of $A_n$ on $v_l$ that starts in an initial state and ends in an accepting state. 
Let $m$ be the number of states of $A_n$. For $0 \leq i \leq l$ let $p_i$ be the state on the run $\pi$ that is reached
after the prefix $w_f^n x y^i$ of $v_l$.

Assume now that $m \leq \lfloor \frac{(1-\alpha) n - c - d}{e} \rfloor$.
Then there must exist numbers $i$ and $j$ with $0 \leq i < j \leq \lfloor \frac{(1-\alpha) n - c - d}{e} \rfloor$
such that $p_i = p_j =: p$. By cutting off cycles at $p$ from the run $\pi$ and repeating this, we finally obtain a run of $A_n$ on 
an input stream  $v_k = w_f^n x y^k z$ with $k \leq \lfloor \frac{(1-\alpha) n - c - d}{e} \rfloor$.
This run still goes from an initial state to an accepting state. Hence, $A_n$ accepts with probability $>0$ an input stream 
$v_k$ with $k \leq \lfloor \frac{(1-\alpha) n - c - d}{e} \rfloor$.  This contradicts our previous observation.
Hence, for every $n \geq d$ with $n \equiv d \pmod e$, $A_n$ must have more than 
$\lfloor \frac{(1-\alpha) n - c - d}{e} \rfloor$ states.
This implies
$$
s_{\mathcal{A}}(n) \geq \log_2\bigg( \frac{(1-\alpha) n - c - d}{e} \bigg) \geq \log_2 n + \log_2 (1-\alpha) - \O(1) ,
$$
which proves the theorem.
\end{proof}

Theorem~\ref{theorem:deterministic_lb} is a direct corollary of Theorem~\ref{theorem-lower-bound} since every deterministic 
sliding window tester is also a nondeterministic sliding window tester.

\begin{ex} \label{ex-lower-bound}
For the lower bound $\log_2 n + \log_2 (1-\eps/\eps_0) - \O(1)$ 
in Theorem~\ref{theorem-lower-bound} the Hamming gap has to be strictly below $\eps_0 n$,
where $\eps_0$ is a constant that depends on $L$.
This is in general not avoidable. Consider for instance the language $L_c = (\{a,b\}^{c-1} a)^*$. 
It is nontrivial, since for any $k$, the word $w_k=b^{c\cdot k}$ has Hamming distance $\dist(w_k,L_c)=k$ from $L_c$.
On the other hand this is also the worst-case, i.e., any word $w$ of length $n = ck$ has Hamming distance $\dist(w,L_c)\le k = n/c$ from $L_c$.
Hence, with constant space one can achieve a Hamming gap of $n/c$ using the algorithm that always accepts.
%
\end{ex}

\subsection{Co-nondeterministic lower bounds}\label{sec:one-sided_lb}

Using a power set construction presented in the following Lemma~\ref{lemma-powerset}, 
one directly obtains from Theorem~\ref{theorem-lower-bound} a lower bound for co-nondeterministic sliding window testers:

\begin{lem}\label{lemma-powerset}
If there exists a co-nondeterministic sliding window tester $\mathcal{A}=(A_n)_{n\ge 0}$ for $L$ with Hamming gap $\gamma(n)$ that uses space $s(n)$, then there is a deterministic sliding window tester for $L$ with Hamming gap $\gamma(n)$ that uses space $2^{s(n)}$.
\end{lem}

\begin{proof}
Let $A_n=(Q_n,\Sigma,I_n,\delta,F_n)$. We apply the powerset construction and transform every coNFA $A_n$ into a DFA $A_n'$ with state set $\mathcal{P}(Q_n)$ (the power set of $Q_n$). The only difference to the powerset construction for NFAs is the following: a state $Q\subseteq Q_n$ of $A_n'$ is final if and only if $Q\subseteq F_n$ (for NFAs it is only required that $Q\cap F_n\neq \emptyset$). It is straightforward to see that $L(A_n)=L(A_n')$. Moreover, $A_n'$ has $2^{|Q_n|}$ many states.
\end{proof}

\begin{thm} \label{theorem-lower-bound-one-sided}
For every non-trivial regular language $L$ there is a constant $\eps_0$, $0 < \eps_0 \le 1$,
such that for every $0 \le \eps < \eps_0$, every co-nondeterministic sliding window tester for $L$ with Hamming gap $\eps n$
uses space  at least $\log_2 \log_2 n-\O(1)$ on an infinite set of window sizes $n$ (that only depends on $L$). 
\end{thm}

Note that a randomized sliding window tester for $L$ with one-sided error is
also a co-nondeterministic sliding window tester for $L$. Hence, the doubly logarithmic space lower bound for non-trivial regular languages from
Theorem~\ref{theorem:one-sided_lb} is a direct corollary of Theorem~\ref{theorem-lower-bound-one-sided}.
Finally, for the logarithmic space lower bound in Theorem~\ref{theorem:one-sided_lb} we need the following two lemmas:

\begin{lem}
	\label{lem:sf-ex}
	Every regular suffix-free language excludes a factor.
\end{lem}

\begin{proof}
Let $B=(Q,\Sigma,F,\delta,q_0)$ be an rDFA for $L$. Since $L$ is suffix-free,
we can assume that there is a single maximal SCC that consists of a single state $q_{\mathit{fail}} \notin F$
(if a maximal SCC would contain a final state, then $L$ would not be suffix-free). We have
$\delta(a,q_{\mathit{fail}})=q_{\mathit{fail}}$ for all $a \in \Sigma$.
We construct a word $w_f \in \Sigma^*$ such that $\delta(p, w_f) = q_{\mathit{fail}}$ for all $p \in Q$.
Let $p_1, \ldots, p_m$ be an enumeration of all states in $Q \setminus \{q_{\mathit{fail}}\}$.
We then construct inductively words $w_0,w_1, \ldots, w_m \in \Sigma^*$ such that 
for all $0 \leq i \leq m$:  $\delta(w_i, p) = q_{\mathit{fail}}$ for all $p \in \{p_1, \ldots, p_i\}$. We start
with $w_0 = \lambda$. Assume that $w_i$ has been constructed for some $i < m$.
There is a word $x$ such that that $\delta( x, \delta(w_i, p_{i+1})) = q_{\mathit{fail}}$. We set $w_{i+1} = xw_i$.
Then $\delta(w_{i+1}, p_{i+1}) = \delta(xw_i, p_{i+1})= q_{\mathit{fail}}$ and 
$\delta(w_{i+1},p_j) = \delta(w_ix,p_j)= \delta(x,q_{\mathit{fail}}) = q_{\mathit{fail}}$ for $1 \leq j \leq i$.
We finally define $w_f = w_m$.
\end{proof}

\begin{lem}
\label{lem:not-sf}
Every regular language $L$ satisfies 
one of the following properties:
\begin{itemize}
\item $L$ is a finite union of regular trivial languages and regular suffix-free languages.
\item $L$ has a restriction $L|_N$ which excludes some factor
and contains $y^*z$ for some $y,z \in \Sigma^*$, $|y| > 0$.
\end{itemize}
\end{lem}

\begin{proof}
	Let $B=(Q,\Sigma,F,\delta,q_0)$ be an rDFA for $L$.
	Let $B_r=(Q,\Sigma,F_r,\delta,q_0)$ where $F_r$ is the set of non-transient final states
	and $B_q = (Q,\Sigma,\{q\},\delta,q_0)$ for $q \in Q$.
	We can decompose $L$ as a union of $L_r = L(B_r)$ and all languages $L(B_q)$
	over all transient states $q \in F$.
	Notice that $L(B_q)$ is suffix-free for all transient $q \in F$
	since any run to $q$ cannot be prolonged to another run to $q$.
	If $L_r$ is trivial, then $L$ satisfies the first property.
	If $L_r$ is nontrivial, then by Lemma~\ref{cuts} and Proposition~\ref{prop-nontriv}
	there exists an arithmetic progression $N = \{a+bn \mid n \in \N\}$ such that
	$L_r|_N$ is infinite and excludes some word $w \in \Sigma^*$ as a factor.
	Let $z \in L_r|_N$ be any word.
	Since $B_r$ reaches some non-transient final state $p$ on input $z$
	there exists a word $y$ which leads from $p$ back to $p$.
	We can ensure that $|y|$ is a multiple of $b$ by replacing $y$ by a suitable power $y^i$.
	Then $y^*z \subseteq L_r|_N \subseteq L|_N$.
	Furthermore since each language $L(B_q)$ excludes some factor $w_q$ by Lemma~\ref{lem:sf-ex}
	the language $L|_N \subseteq L_r|_N \cup \bigcup_q L(B_q)$
	excludes any concatenation of $w$ and all words $w_q$ as a factor.
\end{proof}

\begin{restatable}{theorem-restatable}{lbconondet}
Let $L$ be a regular language that is not a finite union of regular trivial languages and regular suffix-free languages. Then there is a constant $\eps_0$, $0 < \eps_0 \le 1$,
such that for every $0 \le \eps < \eps_0$, every co-nondeterministic sliding window tester for $L$ with Hamming gap $\eps n$
uses space at least $\log_2 n + \log_2 (1-\eps/\eps_0) - \O(1)$ on an infinite set of window sizes $n$ (that only depends on $L$). 
\end{restatable}
\begin{proof}
By Lemma~\ref{lem:not-sf}, $L$ has a restriction $L|_N$ which excludes some factor $w_f$
and contains $y^*z$ for some $y,z \in \Sigma^*$, $|y| > 0$.
Let $c = |w_f| \geq 1$. We set $\eps_0 = 1/c$.
Let $d = |z|$ and $e = |y|$.
Fix a window length $n\in N$ and 
define for $k \geq 0$ the input streams $u_k = w_f^n y^k$ and $v_k = u_k z = w_f^n y^k z$.
Consider a co-nondeterministic sliding window tester $\mathcal{A} = (A_n)_{n \geq 0}$  for $L$ with Hamming gap $\eps n$
for some $\eps < \eps_0$. Let $\alpha = c \eps < 1$ and $r= \lfloor \frac{(1-\alpha) n - c - d}{e} \rfloor$.
If $0 \le k \le r$, then the suffix of $v_k$ of length $n$ 
contains at least
$$
\bigg\lfloor\frac{n - d - e k}{c} \bigg\rfloor \geq 
\bigg\lfloor\frac{n - d - (1-\alpha) n + c + d}{c} \bigg\rfloor = 
\bigg\lfloor\frac{\alpha n + c}{c} \bigg\rfloor = \lfloor \eps n + 1 \rfloor  > \eps n
$$
many disjoint occurrences of $w_f$.
Hence, after reading any of the input streams $v_k$ for 
$0 \le k \le r$, the coNFA $A_n$ has to reject, i.e., there is an $A_n$-run on $v_k$ that starts in an initial state and ends in a non-accepting state.
Consider an $A_n$-run $\pi$ on $v_r$ that goes from an initial state to a non-accepting state.
For $0 \leq i \leq r$ let $p_i$ be the state in $\pi$ that is reached after the prefix $w_f^n y^i$ of $v_r$.
Let now $m$ be the number of states of $A_n$ and assume $m\le r$. There must exist numbers $i$ and $j$ with $0 \leq i < j \leq r$
such that $p_i = p_j =: p$.
It follows that there is an $A_n$-run on $y^{j-i}$ that starts and ends in state $p$.
Using that cycle we can now prolong the run $\pi$, i.e., for all $t\ge 0$ there is an $A_n$-run on $v_{r+(j-i)\cdot t}=w_f^n y^{r+(j-i)\cdot t}z$ 
that starts in an initial state and ends in a non-accepting state.

Assume now that the window size satisfies $n \geq d$ and $n \equiv d \pmod e$. Write $n = d + l e$ for some $l \geq 0$.
Note again that each $n$ with this property satisfies $n\in N$ since the word $y^lz$ belongs to $L|_N$.
We have $l > \lfloor \frac{(1-\alpha) n - c - d}{e} \rfloor = r$. For every $k \ge l$,
the suffix of $v_k = w_f^n y^k z$ of length $n$ is $y^l z \in L$. Therefore $A_n$ accepts $v_k$, i.e., for all $k\ge l$,
every $A_n$-run on $v_k$ that starts in an initial state has to end in an accepting state.
This contradicts our observation that for all $t\ge 0$ there is an $A_n$-run on $v_{r+(j-i)\cdot t}$ that goes from an initial state to a 
non-accepting state. Hence, $A_n$ has at least $r+1 \geq \frac{(1-\alpha) n - c - d}{e}$ states.
It follows that 
$$
s_{\mathcal{A}}(n) \geq  \log_2\bigg( \frac{(1-\alpha) n - c - d}{e} \bigg) \geq \log_2 n + \log_2 (1-\eps/\eps_0) - \O(1).
$$
This proves the theorem.
\end{proof}

\section{Further research}

We gave a complete characterization of the space complexity of sliding window  testers for regular 
languages. A natural open research problem is, whether similar results can be shown for context-free languages:
\begin{itemize}
\item Does every context-free language $L$ has a deterministic sliding window tester with Hamming gap $\eps n$ (or even $\mathcal{O}(1)$) that uses space $\mathcal{O}(\log n)$ (or at least 
space $o(n)$)?
\item Does every context-free language $L$ has a randomized sliding window tester with Hamming gap $\eps n$ (or even $\mathcal{O}(1)$) that uses space $\mathcal{O}(1)$ (or at least 
space $o(n)$)?
\end{itemize}
If the answers to these questions turn out be negative, then one might look at 
deterministic context-free languages or visibly pushdown languages.


\end{document}